\documentclass[a4paper,11pt]{article}
\usepackage[utf8]{inputenc}
\usepackage{amsmath,amssymb,amsthm}
\usepackage{dsfont}
\usepackage[T1]{fontenc}
\usepackage[ngerman,english]{babel}
\usepackage[inline]{enumitem}
\usepackage{graphicx}

\numberwithin{equation}{section}

\DeclareMathOperator{\tr}{Tr}

\newtheorem{theorem}{Theorem}

\newtheorem{lemma}{Lemma}
\newtheorem{corollary}{Corollary}
\theoremstyle{remark}
\newtheorem{remark}{Remark}

\theoremstyle{definition}

\newtheorem{assumption}{Assumption}

\newcommand{\dda}{{\rm d}}
\newcommand{\dd}{\,\dda}
\newcommand{\ddd}[1]{\dd^{#1}}

\newcommand{\lpsi}{{\lambda\psi}}

\begin{document}

\title{Bogolubov--Hartree--Fock theory for strongly interacting
  fermions in the low density limit\footnote{\copyright\ 2015 by the authors. This work may be
   reproduced, in its entirety, for non-commercial purposes.}
}

\author{Gerhard Br\"aunlich$^1$, Christian Hainzl$^2$, and Robert
  Seiringer$^3$
   \\ \\  $^1$Institute for Mathematics, Friedrich-Schiller-University Jena \\
Ernst-Abbe-Platz 2, 07743 Jena, Germany
  \\ \\  $^2$Mathematical Institute, University of T{\"u}bingen \\ Auf der Morgenstelle 10, 72076 T\"ubingen, Germany \\ \\
  $^3$Institute of Science and Technology Austria\\ Am Campus 1, 3400
  Klosterneuburg, Austria}

\date{November 25, 2015}

\maketitle

\begin{abstract}
We consider the Bogolubov--Hartree--Fock functional for a fermionic many-body system with two-body interactions. For suitable interaction potentials that have a strong enough attractive tail in order to allow for two-body bound states, but are otherwise sufficiently repulsive to guarantee stability of the system, we show that in the low-density limit the ground state of this model consists of a Bose--Einstein condensate of fermion pairs. The latter can be described by means of the Gross--Pitaevskii energy functional. 
\end{abstract}

\section{Introduction}
\label{sec:introduction}

We consider a gas of fermions confined in an external trap at zero temperature. The particles interact through a two-body potential $V$
which admits a negative energy bound state. At zero temperature and
low particle densities, this leads to the formation of
diatomic molecules forming a Bose-Einstein condensate.
It was realized in the '80s \cite{Leggett, NRS} that BCS theory
can be adequately applied to such types of tightly bound fermions. It
was pointed out in \cite{melo-randeria, zwerger-1992, Pieri-Strinati,
  randeria} that in the low density limit the macroscopic variations
in the pair density should be well captured by the Gross-Pitaevskii (GP)
equation.  From a mathematical point of view, the emergence of the GP
functional in the low density limit was recently proven in \cite{HS}
for the static case, and the dynamical case was subsequently treated in
\cite{HSch}.  The assumption, that the two-body interaction potential
allows for a bound state plays a crucial role.  In the case of weak
coupling where the potential is not strong enough to form a bound
state, the pairing mechanism may still play an important role for the
macroscopic  behavior of the system, but the separation of paired
particles can be much larger, in this case, than the average particle spacing.  In
fact this is the case in the usual BCS description of superconducting
materials. Close to the critical temperature the
macroscopic variation of the pairs is captured by the Ginzburg-Landau
equation in this case, as pointed out by Gorkov \cite{Gorkov} soon after the
introduction of BCS theory.  The first
mathematical proof of the emergence of Ginzburg-Landau theory from BCS
theory was recently given in \cite{FHSS-micro_ginzburg_landau}, which itself relied on earlier work on
the BCS functional \cite{HHSS,HS-T_C,FHNS2007}.

In the current paper our starting point is the full BCS Hartree-Fock 
functional. 
That is,  we include the direct and exchange 
terms in the interaction energy.  One also finds this functional under the name
Bogolubov-Hartree-Fock (BHF) functional in the literature. The inclusion
of the density-density interaction adds additional difficulties concerning stability of the system. It forces us to restrict  to
systems with a two-body potential $V$ that, on the one hand,  has an attractive tail deep enough to allows for
a bound state and, on the other hand, is sufficiently  repulsive at short distances to guarantee stability.   This is consistent with typically considered 
interaction potentials \cite{Leggett}.

\begin{figure}[ht]
  \centering
  \includegraphics[width=0.3\linewidth]{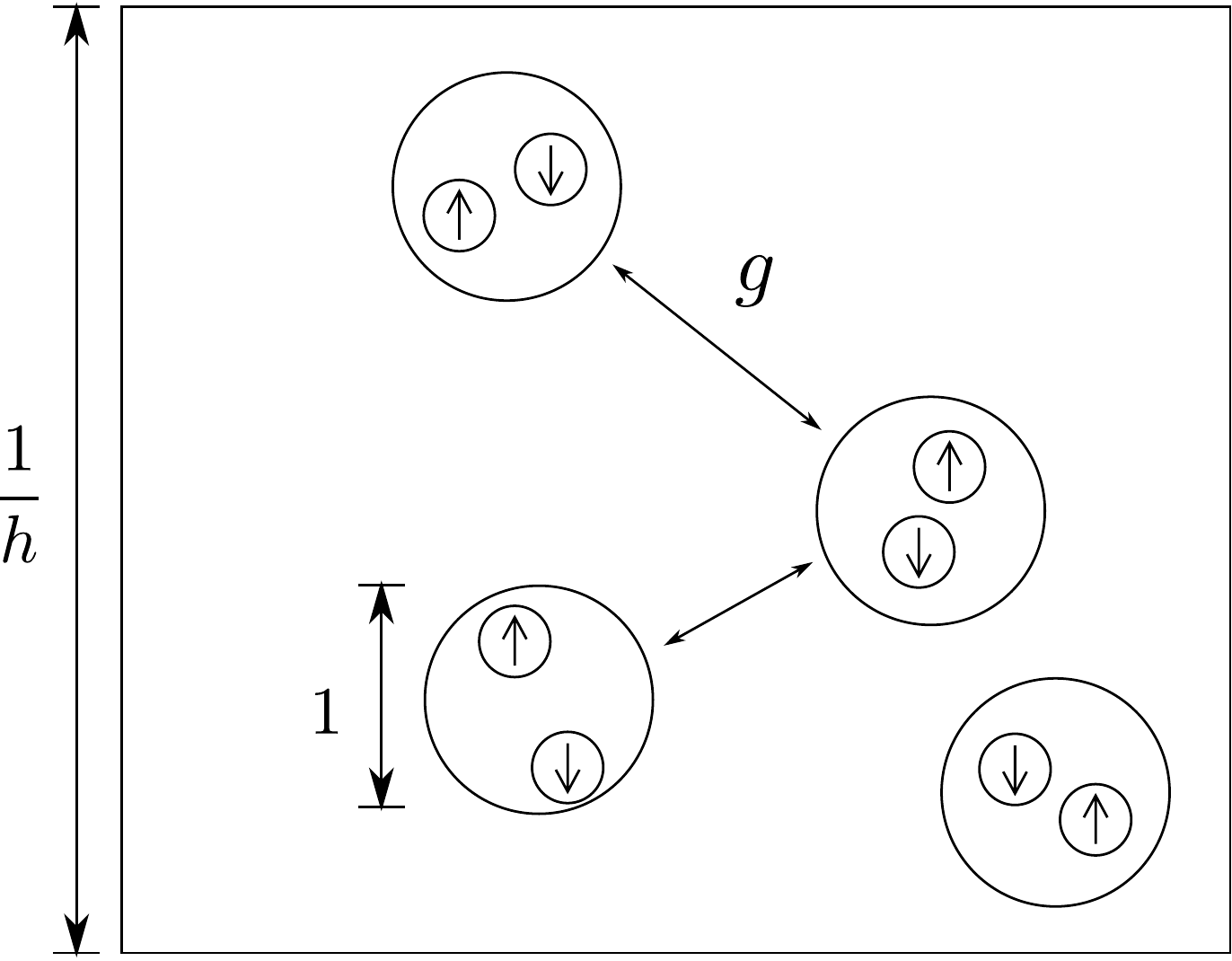}
  \caption{Fermions form diatomic molecules with their 
    repulsive interaction represented by an effective scattering
    length $g$.}
  \label{fig:bosons}
\end{figure}

We shall investigate the ground state energy of the BHF functional in the low density limit.
We introduce a small parameter $h$  playing the role of the inverse particle number, i.e., $N = h^{-1}$. 
We consider external
potentials that  confine the particles on a length scale of order $h^{-1}$, while the range of the interaction among the particles 
is of order one.
This implies that the density of particles is of the order $h^2$.
Hence the small parameter $h$ represents the square root of the
density as well the ratio between the microscopic and macroscopic length scale.
We are going to show that in
the low density limit the fermions group together in pairs, such that
the leading order in the energy is given by the number of pairs,
$1/(2h)$, times the binding energy of a pair, $-E_b$. The next to leading order
is given by the energy of a repulsive Bose gas, consisting of fermion pairs, in a trap, and can be described in terms of  the Gross-Pitaevskii energy  functional.
 More precisely, if $E^{\rm BHF}(h) $
denotes the BHF energy of $1/h$ fermions
we shall obtain
$$ E^{\rm BHF}(h) = - E_b  \frac{1}{2h} + \frac h 2 E^{\rm GP}(g) + O(h^{3/2})$$
for small $h$, where $E^{\rm GP}(g)$
denotes the Gross-Pitaevskii energy with appropriate interaction parameter $g$, which can be computed in terms of the microscopic quantities. The prefactor $h/2$ should be interpreted as $N_{\rm bos}/L^2$, where $N_{\rm bos} = N/2 = 1/(2h)$ is the number of fermion pairs and $L=1/h$ is the macroscopic length scale. 

We will also give a detailed description of the corresponding ground state of the BHF functional. 
Its minimizer turns out to be given, to leading order in $h$, in the form of the two particle wavefunction
$$\alpha(x,y) = h\, \alpha_0\left ({x-y}\right ) \psi\left(h \frac {x+y}2 \right),$$
where $\alpha_0$ is the  ground state of a bound fermion pair with energy $-E_b$,
and $\psi$ solves the GP equation and describes the density
fluctuations of the pairs.

Our work is an extension of \cite{HS} in two directions. First, we
include exchange and direct terms in the energy functional. Second, we avoid working with infinite, periodic systems, which allows us to significantly simplify the proof and also to improve the error bounds, utilizing ideas in \cite{HSch}. In particular, we do not need to use here the rather involved semiclassical estimates of 
\cite{FHSS-micro_ginzburg_landau}.

Our work presents the first proof of the occurrence of pairing
in the ground state of a non-translation invariant
Bogolubov-Hartree-Fock system. (For a translation invariant system this was previously shown in
\cite{BHS}.) The ground
state properties of the BHF functional, in the
context of Newtonian interaction, were studied in
\cite{Lenzmann-Lewin}, see also \cite{BacFroJon-09}. Still it could
not be shown that the fermions in the ground state exhibit pairing. Its occurrence  was  only shown  numerically in \cite{Lewin-Paul}.
In the low density limit, which we are studying here, the
ground state actually predominately consists of pairs, in a sense to be made precise below. In particular, it is essential for our results that the pairing term is included in the energy functional; the Hartree-Fock functional for particle-number conserving states would lead to markedly different results, and is inappropriate for the description of low density gases. 

\section{Main Results}
\label{sec:results}

As in BCS theory, the state of a fermionic system is described by a
self-adjoint operator $\Gamma \in
\mathcal{L}\big(L^2(\mathbb{R}^3)\oplus L^2(\mathbb{R}^3)\big)$,
satisfying $0\leq \Gamma \leq 1$. It is
determined by two operators $\gamma, \alpha \in
\mathcal{L}\big(L^2(\mathbb{R}^3)\big)$ and has the form
\begin{equation*}
  \Gamma =
  \left(
    \begin{matrix}
      \gamma & \alpha \\
      \overline{\alpha} & 1-\overline{\gamma}
    \end{matrix}
  \right),
\end{equation*}
where $0\leq \gamma \leq 1$ is trace class and $\alpha$ is Hilbert-Schmidt and symmetric, i.e. $\alpha(x,y) = \alpha(y,x)$, which implies that $\alpha^* = \bar \alpha$.
We denote by $\overline{\gamma}$, $\overline{\alpha}$ the
operators with kernels $\overline{\gamma(x,y)}$ and
$\overline{\alpha(x,y)}$, respectively. We note that we do not include spin variables here, but rather assume $SU(2)$-invariance of the states \cite{HS-review}.  The full, spin-dependent Cooper-pair wave function is the product of $\alpha$ with an anti-symmetric spin singlet. 
Since $\alpha$ is symmetric, the latter is thus anti-symmetric, as appropriate for fermions.

Given an external potential $W$
and a two-particle interaction potential $V$, the corresponding Bogolubov-Hartree-Fock functional (BHF)  is given by
\begin{equation}
  \label{eq:E_BCS:micro}
  \begin{split}
    \mathcal{E}^{\rm BHF}(\Gamma) &= \tr \big( -\Delta +
    W\big)\gamma + \frac{1}{2} \int_{\mathbb{R}^6} V(x-y)
    |\alpha(x,y)|^2 \ddd{3} x \ddd{3}y\\
    &\quad - \frac{1}{2}\int_{\mathbb{R}^6}|\gamma(x,y)|^2
    V(x-y)\ddd{3}x \ddd{3}y + \int_{\mathbb{R}^6}
    \gamma(x,x)\gamma(y,y) V(x-y)\ddd{3}x \ddd{3}y.
  \end{split}
\end{equation}
We note that the terms in the first line represent the BCS functional, while the last line contains the additional exchange and direct terms in the interaction energy. 
A formal derivation of this functional from quantum
mechanics can be obtained via restriction to quasi-free states, see
\cite{BLS}, \cite[Appendix]{HHSS} or  \cite{HS-review}.
Let us mention that our methods also allow to include a magnetic external vector potential, but for simplicity we shall not 
do so here.

We study a system of $h^{-1}$ fermions interacting by means of a two-body
interaction $V=V(x-y)$, confined in an external  potential of the form $W(hx)$. I.e., the external potential varies on a scale of order $1/h$
whereas $V$ varies on a scale of order one.   Since the trap $W$ confines the particles within
a volume of order $1/h^3$,  the particle density  is of the order 
$h^2$.  Hence the  limit of small $h$
corresponds to a dilute or low density limit.

Since we expect the interaction energy per particle pair to be of the order of the density, we  shall also consider suitably weak external
potentials, i.e., we replace $W$ by $h^2W$. It is convenient
to use macroscopic variables instead of microscopic ones, i.e., we
define $x_h= hx$, $y_h= hy$, $\alpha_h(x,y) =
h^{-3}\alpha(\frac{x}{h},\frac{y}{h})$, and $\gamma_h(x,y) =
h^{-3}\gamma(\frac{x}{h},\frac{y}{h})$.
The resulting BHF functional is then given by (now dropping the subscripts $h$)
\begin{equation}
  \label{eq:E_BCS}
  \begin{split}
    \mathcal{E}^{\rm BHF}(\Gamma) &= \tr (-h^2 \Delta + h^2W)\gamma +
    \frac{1}{2} \int_{\mathbb{R}^6} V\Bigl(\frac{x-y}{h}\Bigr)
    |\alpha(x,y)|^2 \ddd{3} x \ddd{3}y\\
    &\quad - \frac{1}{2}\int_{\mathbb{R}^6}|\gamma(x,y)|^2
    V\Bigl(\frac{x-y}{h}\Bigr)\ddd{3}x \ddd{3}y + \int_{\mathbb{R}^6}
    \gamma(x,x)\gamma(y,y) V\Bigl(\frac{x-y}{h}\Bigr)\ddd{3}x
    \ddd{3}y,
  \end{split}
\end{equation}
where $W=W(x)$ is independent of $h$. The corresponding ground state energy is denoted as 
\begin{equation}\label{ebhf} 
  E^{\rm BHF}(h) = \inf \{\mathcal{E}^{\rm BHF}(\Gamma) \, | \, 0 \leq \Gamma \leq 1, \, \tr \gamma = 1/h\}\,.
\end{equation}

\begin{figure}[ht]
  \centering
  \includegraphics[width=0.2\linewidth]{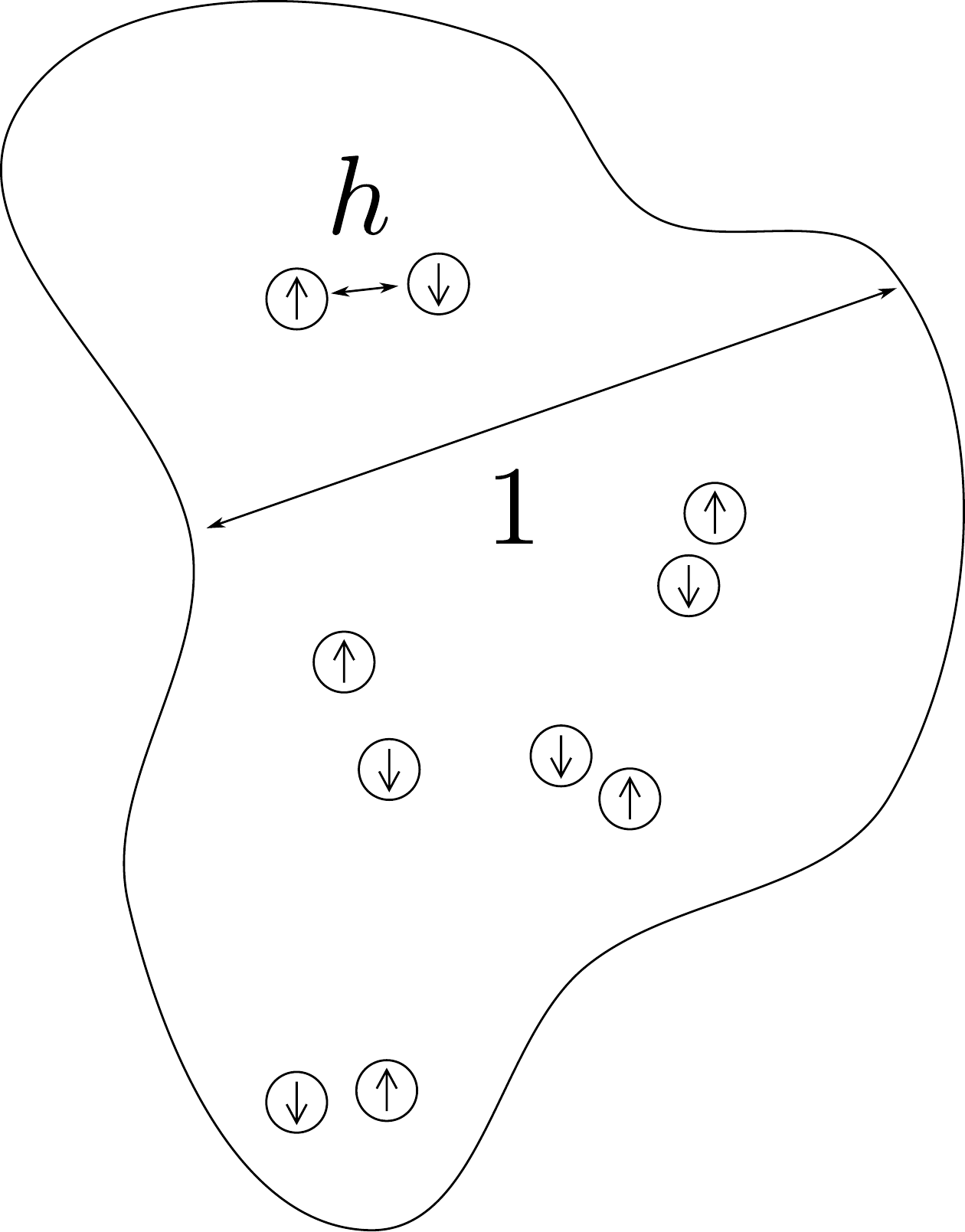}
  \caption{Separation of scales: The range of the interaction between the fermions
    is of order $h$, while the external potential varies on a scale of
    order $1$.}
  \label{fig:system}
\end{figure}

For $\psi \in H^1(\mathbb{R}^3)$ the GP
functional is defined as
\begin{equation}
  \label{eq:E_GP}
  \mathcal{E}^{\rm GP}(\psi)
  =
  \int_{\mathbb{R}^3}\left( \frac{1}{2} |\nabla\psi(x)|^2 + 2
    W(x)|\psi(x)|^2 + g|\psi(x)|^4 \right) \ddd{3}x\,.
\end{equation}
The factors $1/2$ and $2$, respectively, in the first two terms result from the fact that \eqref{eq:E_GP} describes fermion {\em pairs}. 
The  interaction parameter $g>0$ will be determined by the BHF functional and
represents the interaction strength among different pairs.     We
denote the ground state energy of the GP functional as
\begin{equation}
  E^{\rm GP}(g) = \inf \{  \mathcal{E}^{\rm GP}(\psi) \, | \, \psi \in H^1(\mathbb{R}^3), \, \|\psi\|_2^2 = 1\}.
\end{equation}

We shall consider the minimization problem \eqref{ebhf}  and show
that its value in the limit $h\to 0$ is to leading order given by the
binding energy of the fermion pairs, i.e. $-E_b \frac 1{2h}$.
This assumes, of course, that  the
two-body interaction potential $V$ allows for  a negative energy bound state, which is 
part of the following assumption.

\begin{assumption}
  \label{asm:V}
  Let $V\in  L^\infty(\mathbb{R}^3)$ be real-valued, with $V(x)
  = V(-x)$, such that $|\, \cdot \,| V(\,\cdot\,) \in L^1(\mathbb{R}^3)$ and $-2\Delta +V$ has a normalized ground state $\alpha_0$
   with corresponding ground state energy
  $-E_b < 0$.
\end{assumption}

Including direct and exchange term into the BCS functional gives rise
to a new problem. A priori it is not clear whether the functional guarantees stability of the second kind.
To ensure it  we impose
the following further assumption on $V$.
\begin{assumption}
  \label{asm:V:dir_ex}
  There is $U\in L^1(\mathbb{R}^3)\cap L^\infty(\mathbb{R}^3)$, with
  non-negative Fourier transform
  $\widehat{U} \geq 0$, such that $V - \frac 12 V_+ \geq U$. Here $V_+
  = \frac{1}{2}(|V|+V)$ denotes the positive part of $V$.
\end{assumption}
In other words,  we consider potentials which, after cutting its
positive part in half, can be bounded from below
by functions with a non-negative Fourier transform. In particular, this means that
the potentials have to  have a strong enough
repulsive core and a relatively small attractive tail, which still has to be large enough to
allow for bound states.

\begin{remark}
  The following construction shows that it is easy to find potentials
  $V$ with the desired properties of Assumptions~\ref{asm:V} and~\ref{asm:V:dir_ex}:
  Choose a potential $U$ which is strictly negative on an open set
  $\Omega \subset \mathbb{R}^3$, such that $\widehat{U} \geq 0$.  The
  latter property can be ensured, e.g., by taking $U$ to be the convolution of some function $u$ with its reflection $u(-\,\cdot\,)$.   Now set $V =
  2U_+ - U_-$. Obviously this $V$ fulfills
  Assumption~\ref{asm:V:dir_ex}.  Finally, scale $V$ according to $V
  \mapsto \lambda V$ until the negative part is deep enough for a bound state to appear. 
  \end{remark}

With these assumptions we are ready to formulate our main theorem.

\begin{theorem}
  \label{thm:GP}
  Let $W\in L^\infty(\mathbb{R}^3)$ be real-valued.  Under
  Assumptions~\ref{asm:V} and \ref{asm:V:dir_ex}, we have for small
  $h$,
  \begin{equation}
    \label{eq:thm:GP}
    E^{\rm BHF}(h)
    = -
    {E_b} \frac 1 {2h}
    + \frac h 2 E^{\rm GP}(g) + O(h^{3/2}),
  \end{equation}
  where $g>0$ is given by
  \begin{equation*}
    g = (2\pi)^3 \int_{\mathbb{R}^3} |\widehat{\alpha}_0(p)|^4(2p^2+E_b)\ddd{3}p
    - \int_{\mathbb{R}^3} |(\overline{\alpha_0} *
    \alpha_0)(x)|^2\,V(x) \ddd{3}x
    + 2\int_{\mathbb{R}^3} V(x) \ddd{3}x.
  \end{equation*}
  Moreover, if $\Gamma$ is an approximate minimizer of
  $\mathcal{E}^{\rm BHF}$, in the sense that
  \begin{equation}\label{apprmin}
    \mathcal{E}^{\rm BHF}(\Gamma)
    \leq -{E_b}\frac{1}{2h} + \frac h 2\left(E^{\rm GP}(g) + \epsilon\right)
  \end{equation}
  for some $\epsilon > 0$, then the corresponding $\alpha$ can be
  decomposed as
  \begin{equation}
    \label{eq:alpha_psi:0}
    \alpha = \alpha_\psi + \xi, \qquad \|\xi\|^2_2 \leq O(h),
    \qquad \|\alpha\|^2_2 = O(h^{-1}),
  \end{equation}
  where
  \begin{equation}
    \label{eq:alpha_psi}
    \begin{split}
      \alpha_\psi(x,y) &= h^{-2}
      \psi\Bigl(\frac{x+y}{2}\Bigr)\alpha_0\Bigl(\frac{x-y}{h}\Bigr),
    \end{split}
  \end{equation}
  and $\psi$ is an approximate minimizer of $\mathcal{E}^{\rm GP}$ in
  the sense that
  \begin{equation}
    \label{eq:approx_minimizer}
    \mathcal{E}^{\rm GP}(\psi)
    \leq E^{\rm GP}(g) + \epsilon + O(h^{1/2}).
  \end{equation}
\end{theorem}

\begin{remark}\label{gis}
  In contrast to the case of the  usual BCS functional \cite{HS,HSch},  where the coupling constant $g$ only consists only of the
  BCS term 
  \begin{equation}\label{def:gbcs}
  g_{\rm bcs} = (2\pi)^3 \int_{\mathbb{R}^3}
  |\widehat{\alpha}_0(p)|^4(2p^2+E_b)\ddd{3}p\,,
  \end{equation}
  it receives here two
  additional contributions from the direct and exchange energies, 
  \begin{equation*}
    g_{\rm dir} = 2 \int_{\mathbb{R}^3} V(x) \ddd{3}x\quad\textrm{and }\quad
    g_{\rm ex} = - \int_{\mathbb{R}^3} |({\alpha_0} * \alpha_0)(x)|^2\,V(x) \ddd{3}x\,,
  \end{equation*}
  respectively. It is easy to see that our Assumption~\ref{asm:V:dir_ex} implies that $g_{\rm dir} + g_{\rm ex} \geq 0$, hence $g>0$. 
  \end{remark}

\begin{remark}
  The proof of Thm.~\ref{thm:GP} partly relies on ideas in \cite{HSch}, where the corresponding time-dependent problem was studied for the BCS functional, i.e., in the absence of direct and exchange term.
  A similar result can also be shown to hold  in the case of the
  time-dependent BHF equation, which in a different context  was
  studied in \cite{HaiLenLewSch-10}.  By following the strategy of
  \cite{HSch} and handling the exchange and direct terms in a similar way as done here, one can derive 
   the time-dependent GP equation with interaction 
  parameter $g$.
\end{remark}

{\em Notation}: In the following we often write $a \lesssim b$ to denote $a \leq Cb$ for some generic constant $C>0$. 

\section{Stability}
\label{sec:stability}

Before giving a sketch of the proof of Theorem~\ref{thm:GP} we show how
Assumption~\ref{asm:V:dir_ex} gives rise to stability of the second kind. In fact we simply show
that the assumption guarantees that the sum of the direct and exchange terms is 
non-negative.  To this aim we first consider the exchange term and
estimate
\begin{align*}
  - \int_{\mathbb{R}^6}|\gamma(x,y)|^2 V\bigl((x-y)/h\bigr)\ddd{3}x
  \ddd{3}y &\geq -\int_{\mathbb{R}^6}|\gamma(x,y)|^2
  V_+\bigl((x-y)/h\bigr)\ddd{3}x \ddd{3}y \\
  & \geq - \int_{\mathbb{R}^6}\gamma(x,x)\gamma(y,y)
  V_+\bigl((x-y)/h\bigr)\ddd{3}x \ddd{3}y\,,
\end{align*}
using $|\gamma(x,y)|^2 \leq \gamma(x,x)\gamma(y,y)$.  Hence we
have for the sum of direct and exchange term
\begin{align*}
  2\int_{\mathbb{R}^6} \gamma&(x,x)\gamma(y,y)
  V\bigl((x-y)/h\bigr)\ddd{3}x \ddd{3}y -
  \int_{\mathbb{R}^6}|\gamma(x,y)|^2
  V\bigl((x-y)/h\bigr)\ddd{3}x \ddd{3}y\\
  &\geq 2\int_{\mathbb{R}^6} \gamma(x,x)\gamma(y,y)
  \left (V- V_+/2\right)\bigl((x-y)/h\bigr)\ddd{3}x \ddd{3}y \\
  &\geq 2 \int_{\mathbb{R}^6}  \gamma(x,x)\gamma(y,y) U\bigl( (x-y)/h \bigr)  \ddd{3}x \ddd{3}y
\end{align*}
where we used the assumption $ \left (V- \frac 12 V_+\right) \geq U$. Since $\widehat{U} \geq
0$ the last term is non-negative. Hence the question of stability is reduced to the corresponding problem for the BCS functional, and is easily seen to hold under our assumptions on $V$.

\section{Sketch of the proof of Theorem~\ref{thm:GP}}
\label{sec:proofs}

The proof of \eqref{eq:thm:GP} consists of deriving appropriate upper and lower bounds on $E^{\rm BHF}(h)$. 

\subsection{Upper bound}

For the upper bound one has to construct  a suitable trial state. We shall proceed similarly to \cite{HSch} and define
the trial state $\Gamma_\psi$ via the pair wavefunction
\begin{equation}\label{def:alphapsi}
\alpha_\psi(x,y) = h^{-2} \psi\Bigl(\frac{x+y}{2}\Bigr)\alpha_0\Bigl(\frac{x-y}{h}\Bigr)\,.
\end{equation}
Since we expect that the system in its ground state energy consists
predominantly of pairs we define the one particle density
$\gamma_\psi$ such that to leading order it equals $\overline
{\alpha_\psi} \alpha_\psi$. 
More precisely, we choose the
trial state
\begin{equation}\label{eq:Gamma_trial0}
\Gamma_\psi = \left(
    \begin{matrix}
      \gamma_\psi & \alpha_\psi \\
      \overline{\alpha_\psi} & 1-\overline{\gamma_\psi}
    \end{matrix}\right)
 \end{equation}
    such that
  \begin{equation}
    \label{eq:Gamma_trial}
      \gamma_\psi = \alpha_\psi\overline{\alpha_\psi} +
      (1+ h^{1/2})\alpha_\psi\overline{\alpha_\psi}\alpha_\psi\overline{\alpha_\psi} \,.
      \end{equation}
 The function $\psi$ here is only approximately normalized, i.e., $\|\psi\|_2 = 1 + O(h^2)$, to ensure that $ \tr \gamma_\psi = 1/h$. 
   We will see below that for small enough $h$ the definition \eqref{eq:Gamma_trial}
  guarantees that $0\leq \Gamma_\psi \leq 1$.

  In the limit of small $h$ the GP energy functional  emerges from the BHF functional $\mathcal{E}^{\rm
    BHF}(\Gamma_\psi)$  as follows.
  If we  consider  the kinetic energy term plus the pairing term   and subtract  the total binding
  energy,
  $-\frac{E_b}{2} \frac{1}{h} =
  -\frac{E_b}{2}\tr{\gamma_\psi}$, the contribution to 
  \begin{equation}
 \tr
    \left(-h^2 \Delta + E_b/2\right)\gamma_\psi +
    \frac{1}{2}\int_{\mathbb{R}^6}V\Bigl(\frac{x-y}{h}\Bigr)
    |\alpha_\psi(x,y)|^2 \ddd{3}x \ddd{3}y
  \end{equation}
  coming from the  $\alpha_\psi \overline{\alpha_\psi}$ term in
  \eqref{eq:Gamma_trial} can be written as
  \begin{equation*}
    \int_{\mathbb{R}^3} \left \langle \alpha_\psi(\cdot, y), \left[-h^2\Delta_x
        +
        \frac{1}{2}V\Bigl(\frac{\cdot -y}{h}\Bigr)
        +\frac{E_b}{2}\right] \alpha_\psi(\cdot ,y)\right \rangle
    \ddd{3}y.
  \end{equation*}
  Since $\alpha_\psi(x,y)$ is symmetric we can replace $\Delta_x $ by
  $\frac 12 (\Delta_x + \Delta_y)$.  In terms of center of mass
  $X=(x+y)/2$ and relative coordinates $r = x-y$ the kinetic energy has the form
  $\Delta_x + \Delta_y = \frac{1}{2}\Delta_X + 2\Delta_r$, such that
  in these coordinates the last term has the form
   \begin{align}\nonumber 
    & h^{-4}
     \int_{\mathbb{R}^6}\overline{\alpha_0(r/h)\psi(X)}\left[-\frac{h^2}{4}\Delta_X
       -h^2 \Delta_r+ \frac{1}{2}V(r/h) +E_b/2\right]\alpha_0(r/h)
     \psi(X) \ddd{3}X \ddd{3}r \\  & =  \frac h4 \int_{\mathbb{R}^3}|\nabla
     \psi(x)|^2 \ddd{3}x,\label{nablaterm}
   \end{align}
   where we used the fact that  $\alpha_0$ is the normalized ground state of $-\Delta + \frac 12 V$. 
   
  The term
  $\alpha_\psi\overline{\alpha_\psi}\alpha_\psi\overline{\alpha_\psi}$ of $\gamma_\psi$
  inserted into
   $$\tr [-h^2 \Delta +
   E_b/2] \gamma_\psi$$ contributes  the quartic term $\frac h 2 g_{\rm bcs} 
   \int_{\mathbb{R}^3} |\psi(x)|^4 \ddd{3}x$ term in the GP
   functional. The remaining part of the $\frac h 2g \int_{\mathbb{R}^3}
   |\psi(x)|^4 \ddd{3}x$ term is due to the contribution of  $\alpha_\psi\overline{\alpha_\psi}$ in the direct and
   exchange interaction terms. The estimation of these terms is straightforward but tedious and occupies the main part of the proof.

  Furthermore, it will be easy to show that
  $$ h^2 \tr W \alpha_\psi\overline{\alpha_\psi} =  h^{-2} \int_{\mathbb{R}^3} W(X+  r/2) | \psi(X)|^2 |\alpha_0(r/h)|^2  \ddd{3}r
  = h \int_{\mathbb{R}^3} W(X) |\psi(X)|^2 \ddd{3}X+ O(h^2).$$
  Consequently we shall obtain
  \begin{equation}
    \label{eq:BHFequGP}
    \mathcal{E}^{\rm BHF}(\Gamma_{\psi}) + E_b \frac 1{2h}   = \frac h 2 \mathcal{E}^{\rm GP}(\psi) + O(h^{3/2}).
  \end{equation}
  Finally, we remark that the constraint $\tr\gamma_\psi = 1/h$
  implies for $\psi$ that $ \| \psi\|_2^2 = (1 - O(h^2))$. 
 Since
 $$\big|\mathcal{E}^{\rm GP}(\psi) - \mathcal{E}^{\rm GP}\big([1+O(h^2)]\psi\big)\big| \leq O(h^2)$$
 we obtain the bound
 \begin{equation}
   \label{eq:upper_bound}
   \inf_{\substack{0 \leq \Gamma\leq 1 \\ \tr(\gamma) = 1/h}}
   \mathcal{E}^{\rm BHF}(\Gamma)+E_b \frac{1}{2h}
   \leq \frac h2 \inf_{\substack{\psi \in
       H^1(\mathbb{R}^3)\\\|\psi\|_2^2 = 1}}\mathcal{E}^{\rm GP}(\psi)
   + O(h^{3/2})\,.
 \end{equation}
 The precise derivation of this bound will be given in Section~\ref{sec:upper_bound}.
 
 \begin{remark}
   \hfill
   \begin{itemize}
   \item Since the infimum of $\mathcal{E}^{\rm BHF}$ is attained by a projection \cite{BLS}, it would be natural to chose the trial state $\Gamma_\psi$ as a projection. The operator
     $\gamma_\psi$ would then satisfy
     $\gamma_\psi = \gamma_\psi^2 + \alpha_\psi\overline{\alpha_\psi}$.
     The expansion of $\gamma_\psi$ in terms of  $ \alpha_\psi\overline{\alpha_\psi}$ would be more complicated, however, and we find the choice \eqref{eq:Gamma_trial} more convenient. 
   \item Our trial state satisfies $0\leq
     \Gamma_\psi \leq 1$ for small enough $h$.  To see this note that the condition is equivalent to $0 \leq
     \Gamma_\psi(1-\Gamma_\psi)$.  If $\gamma_\psi$ is of the special
     form \eqref{eq:Gamma_trial}, which is a function of
     $\alpha_\psi\overline{\alpha_\psi}$, the off-diagonals of
     \begin{align*}
       \Gamma_\psi(1-\Gamma_\psi) &= \left(
         \begin{array}{ll}
           \gamma_\psi - \gamma_\psi^2 -
           \alpha_\psi\overline{\alpha_\psi} &
           \alpha_\psi\overline{\gamma_\psi} - \gamma_\psi\overline{\alpha_\psi } \\
           \overline{\gamma_\psi}\alpha_\psi -
           \overline{\alpha_\psi}\gamma_\psi &
           \overline{\gamma_\psi} - \overline{\gamma_\psi}^2 -
           \overline{\alpha_\psi}\alpha_\psi
         \end{array}\right)\\
       &= \left(
         \begin{array}{ll}
           \gamma_\psi - \gamma_\psi^2 -
           \alpha_\psi\overline{\alpha_\psi} &
           0 \\
           0 &
           \overline{\gamma_\psi} - \overline{\gamma_\psi}^2 -
           \overline{\alpha_\psi}\alpha_\psi
         \end{array}\right)
     \end{align*}
     vanish and thus the statement is equivalent to
     \begin{equation}
       \label{eq:Gamma_positivity}
       \gamma_\psi - \gamma_\psi^2 -
       \alpha_\psi\overline{\alpha_\psi} \geq 0.
     \end{equation}
     Plugging in the expression  \eqref{eq:Gamma_trial} for $\gamma_\psi$
     \eqref{eq:Gamma_positivity} is equivalent to
     \begin{equation}\label{guara}
       \alpha_\psi\overline{\alpha_\psi}
       \big(h^{1/2}
       -2(1+h^{1/2})\alpha_\psi\overline{\alpha_\psi}
       -(1+h^{1/2})^2(\alpha_\psi\overline{\alpha_\psi})^2\big)\alpha_\psi\overline{\alpha_\psi}
       \geq 0.
     \end{equation}
     In Corollary~\ref{cor:ineqgamma} below we shall show that the operator norm of $\alpha_\psi$ satisfies  $\|\alpha_\psi \|_\infty \lesssim h^{1/2} $, which guarantees that \eqref{guara}  is  satisfied for $h$ small enough.  In fact, $h^{1/2}$ in \eqref{eq:Gamma_trial} could be replaced by any factor large compared to $h$, but a different choice would not improve our error  bounds. 
      \end{itemize}
 \end{remark}

 \subsection{Lower bound}

 From the upper bound we learn that for an approximate ground state
 $\Gamma$ we can assume
$$\mathcal{E}^{\rm
  BHF}(\Gamma) \leq -E_b \frac{1}{2h} + O(h).$$ We will show in
Lemma~\ref{lemma:apriori} that the corresponding
$\alpha$ necessarily has to be of the form
\begin{equation}\label{deco}
\alpha(x,y) = \alpha_\psi(x,y) + \xi(x,y) = h^{-2}\psi\Bigl(\frac{x+y}{2}\Bigr)\alpha_0\Bigl(\frac{x-y}{h}\Bigr)  + \xi(x,y)
\end{equation}
 for an appropriate
$\psi \in H^1(\mathbb{R}^3)$, with $\xi$ being small compared  to $\alpha_\psi$, i.e.,
$$\|\xi\|^2_2 \lesssim h^2 \|\alpha_\psi\|^2_2 = O(h) .$$
The function $\psi$ is obtained by projecting $\alpha$ in the
direction of $\alpha_0$ with respect to the relative coordinates, 
$$\psi(X) = \frac{1}{h}\int_{\mathbb{R}^3}
\alpha_0(r/h)\alpha(X+r/2,X-r/2) \ddd{3}r.$$
We shall show that $\|\psi\|_2 = 1 + O(h^2)$. 

With this $\psi$ at hand one can then define $\Gamma_\psi$ as in \eqref{eq:Gamma_trial0}--\eqref{eq:Gamma_trial}. With the help of the decomposition \eqref{deco} 
one then argues  that the difference between
  $\mathcal{E}^{\rm BHF}(\Gamma)$ and $\mathcal{E}^{\rm
    BHF}(\Gamma_\psi)$ is bounded from below by a term of higher order than
  the contribution from the GP functional. More precisely, 
  \begin{equation}
    \label{eq:E:diff}
    \mathcal{E}^{\rm
      BHF}(\Gamma) \geq \mathcal{E}^{\rm
      BHF}(\Gamma_\psi)  - O(h^{3/2}).
  \end{equation}
This estimate is uniform in $\psi$, since it is possible to obtain a priori bounds on the $H^1$-norm of  $\psi$ that are independent of $h$. 
  Using now our calculation \eqref{eq:BHFequGP} from the upper bound
  immediately implies
  \begin{equation}
    \label{eq:lower_bound}
    \inf_{\substack{0 \leq \Gamma\leq 1 \\ \tr(\gamma) =
        1/h}}\mathcal{E}^{\rm BHF}(\Gamma)
    +E_b\frac{1}{2h}
    \geq  \frac h 2 \inf_{\substack{\psi \in
        H^1(\mathbb{R}^3)\\\|\psi\|_2^2 = 1}}
    \mathcal{E}^{\rm GP}(\psi) - O(h^{3/2}).
  \end{equation}
  Together with \eqref{eq:upper_bound} this combines to
  \eqref{eq:thm:GP}.

  \section{Useful properties of the pair-wavefunction}

  In the following we shall derive some useful properties of the pair wavefunction \eqref{eq:Gamma_trial}, which will be used throughout the proof. 
  Recall that $\alpha_0$ was defined in Assumption \ref{asm:V}
  to be the normalized ground state of $-2\Delta +V$. It is a rapidly decaying function, and both $|\alpha_0|$ and $|\nabla \alpha_0|$ have smooth Fourier transforms which are  in $L^p(\mathbb{R}^3)$ for any $p\geq 2$. 
  
  \begin{lemma}
    \label{lemma:inequalpha}
    Let $\alpha_\psi $ be defined as in \eqref{eq:Gamma_trial}, with $\psi \in
    H^1(\mathbb{R}^3)$. 
    \begin{enumerate}[label=(\roman*)]
    \item For $n \in \{2,4,6\}$, 
          \begin{subequations}
        \begin{align}
          \label{eq:alpha_n}
          \|\alpha_\psi\|_n^n \lesssim  h^{n-3}\|\psi\|_n^n\,
          \|\widehat{|\alpha_0|}\|_n^n, \\
          \label{eq:nabla_alpha_n}
          \|\nabla_{(x-y)}\alpha_\psi\|_n^n \lesssim
          h^{-3}\|\psi\|_n^n\, \|\widehat{|\nabla \alpha_0|}\|_n^n,
        \end{align}
      \end{subequations}
      where
    $$(\nabla_{(x-y)} \alpha_\psi)(x, y) = h^{-3}
    \psi\bigl((x+y)/2\bigr) (\nabla\alpha_0)\bigl((x-y)/h\bigr).$$
  \item With $g_{\rm bcs}$ defined in \eqref{def:gbcs}, 
    \begin{equation}
      \label{eq:Delta_alpha_psi_4}
      \tr \big( (-h^2\Delta+E_b/2)
      \alpha_\psi\overline{\alpha_\psi}\alpha_\psi\overline{\alpha_\psi}
      \big) = \frac h 2  g_{\rm bcs} \|\psi\|_4^4 + O(h^2) \|\nabla\psi\|_2^4 \,.
    \end{equation}
  \item \hfill\makebox[0pt][r]{\begin{minipage}[b]{\textwidth}
        \begin{equation}
          \label{eq:sup_alpha2}
          \| \alpha_\psi \overline{\alpha_\psi}(\cdot,\cdot)\|_\infty
          = \sup_x | \alpha_\psi \overline{\alpha_\psi}(x,x)| \lesssim
          h^{-2} \|\alpha_0\|_3^2\,\|\nabla \psi\|_2^2.
        \end{equation}
      \end{minipage}}
  \item Let $\sigma$ be a Hilbert-Schmidt operator. Then
    \begin{equation}
      \label{twonorm}
      |\sigma \alpha_\psi (x,x)| \lesssim h^{-1}  \| \sigma (\cdot, x)\|_2
      \|\nabla \psi\|_2\, \|\alpha_0\|_3 \quad \forall x \in \mathbb{R}^3\,.
    \end{equation}
  \end{enumerate}
\end{lemma}
Let us mention that we use the symbol $\| \cdot\|_p$ for the
$L^p$-norm of functions as well as for the operator norm in the
corresponding Schatten class. E.g.,  the left side of \eqref{eq:alpha_n} concerns Schatten norms, while on the right side the norms are in $L^n(\mathbb{R}^3)$.

\begin{proof}[Proof of Lemma~\ref{lemma:inequalpha}, Part I]
  We postpone the proof of \eqref{eq:alpha_n},
  \eqref{eq:nabla_alpha_n} and \eqref{eq:Delta_alpha_psi_4} to 
  the appendix.  In order to see
  \eqref{eq:sup_alpha2} we use the H\"older  and Sobolev inequalities as 
  \begin{align*}
    \begin{split}
      (\alpha_\psi\overline{\alpha_\psi})(x,x) &= \int_{\mathbb{R}^3}
      |\alpha_\psi(x,y)|^2 \ddd{3}y = h^{-4}\int_{\mathbb{R}^3}
      \left|\alpha_0\big((x-y)/h\big)\right|^2\,\left|\psi\big((x+y)/{2})\big)\right|^2
      \ddd{3}y\\
      &\lesssim h^{-4} \|\alpha_0(\cdot/h)^2\|_{3/2} \| |\psi|^2 \|_3 \lesssim
      h^{-2}  \|\alpha_0\|_3^2\|\nabla \psi\|_2^2 \,.
    \end{split}
  \end{align*}
Similarly, 
  \begin{align*}
    | (\sigma\alpha_\psi)(x,x)| &=h^{-2} \left | \int_{\mathbb{R}^6}
      \overline{\sigma(x,y)} \alpha_0\big((x-y)/h\big)
      \psi\big((x+y)/2\big) \ddd{3}y \right| \\
    &\lesssim h^{-2} \| \sigma (\cdot, x)\|_2 \|\alpha_0(\cdot /h)\|_3
    \|\psi\|_6 \lesssim h^{-1} \| \sigma (\cdot, x)\|_2 \|\nabla \psi\|_2
    \|\alpha_0\|_3\,,
  \end{align*}
which implies \eqref{twonorm}.
\end{proof}

\begin{corollary}
  \label{cor:ineqgamma}
  Let the assumptions be as in Lemma~\ref{lemma:inequalpha} and assume further that $\|\psi\|_{H^1} \lesssim 1$. Then
  \begin{subequations}
    \begin{align}
      \label{eq:alpha_4}
      \| \alpha_\psi \|_4^4 &\lesssim h\|\psi\|_2\|\nabla
      \psi\|_2^3\, \|
      \widehat{\alpha_0}\|_{4}^4 \lesssim h,\\
      \label{eq:alpha_6}
      \| \alpha_\psi \|_6^6 &\lesssim
      h^3  \|\nabla \psi\|_2^6\, \| \widehat{\alpha_0}\|_6^6 \lesssim h^3,\\
      \label{eq:nabla_alpha_6}
      \| \nabla_{(x-y)} \alpha_\psi \|_6^6 &\lesssim h^{-3}\|\nabla
      \psi\|_6^6\, \|\widehat{ |\nabla\alpha_0| }\|_6^6 \lesssim h^{-3},\\
      \label{eq:alpha_inf}
      \|\alpha_\psi\|_\infty &\lesssim  h^{1/2}  \|\nabla \psi\|_2\, \|
      \widehat{\alpha_0}\|_6 \lesssim h^{1/2},\\
      \label{alph4}
   \sup_{x\in\mathbb{R}^3}     (\alpha_\psi\overline{\alpha_\psi}\alpha_\psi\overline{\alpha_\psi})(x,x)
      &\leq \|\alpha_\psi\|^2_\infty  \sup_{x\in\mathbb{R}^3} 
      (\alpha_\psi\overline{\alpha_\psi})(x,x)   \lesssim h^{-1}.
    \end{align}
  \end{subequations}
  Moreover, with $\gamma_\psi$ defined as in \eqref{eq:Gamma_trial}, 
  \begin{subequations}
    \begin{align}
      \label{normgammapsi}
      \|\gamma_\psi\|_\infty &\leq \|\alpha_\psi\|_\infty^2 + (1 +
      h^{1/2} )\|\alpha_\psi\|_\infty^4 \lesssim h  \\
      \label{eq:gamma_psi_diag}
  \sup_{x\in\mathbb{R}^3}    \gamma_\psi(x,x) &\lesssim h^{-2}.
    \end{align}
  \end{subequations}
\end{corollary}
\begin{proof}
  The estimates \eqref{eq:alpha_4}--\eqref{eq:nabla_alpha_6} are a consequence of \eqref{eq:alpha_n} and
  \eqref{eq:nabla_alpha_n}. In the case of $n=6$, we use the Sobolev's
  inequality and in the
  case of $n=4$, we use H\"older combined  with Sobolev
  to conclude
  \begin{equation*}
    \|\psi\|_4 \leq \|\psi\|_2^{1/4} \|\psi\|_6^{3/4} \lesssim
    \|\psi\|_2^{1/4}\|\nabla \psi\|_2^{3/4}.
  \end{equation*}
  Inequality~\eqref{eq:alpha_inf} follows immediately from
  $\|\alpha_\psi\|_\infty \leq \|\alpha_\psi\|_6$ together with
  \eqref{eq:alpha_6}.

It is easy to see that
  $$
  (\alpha_\psi\overline{\alpha_\psi}\alpha_\psi\overline{\alpha_\psi})(x,x)
  \leq \|\alpha_\psi\overline{\alpha_\psi}\|_\infty
  (\alpha_\psi\overline{\alpha_\psi})(x,x)
  $$
  which implies \eqref{alph4} with the use of \eqref{eq:sup_alpha2}.
Eq.~\eqref{normgammapsi} follows from \eqref{eq:alpha_inf} and  
  Eq.~\eqref{eq:gamma_psi_diag} is a direct consequence of \eqref{eq:sup_alpha2} and \eqref{alph4}.
\end{proof}

\begin{remark}
Since $\gamma_\psi$ is to leading order equal to $\alpha_\psi
\overline{ \alpha_\psi}$, we obtain as a corollary that the operator norm of
$\gamma_\psi $ is at most $O(h)$, meaning that the largest eigenvalue of the one-particle density matrix 
is of order $h$.  However, the two-body density matrix corresponding to the state $\Gamma_\psi$ is to leading order of the form
$  | \alpha_\psi \rangle \langle \alpha_\psi | $, and hence has  one large
eigenvalue of order $h^{-1}$. This is a manifestation of the Bose--Einstein condensation of the fermion pairs. 
\end{remark}

\section{Upper bound}

\label{sec:upper_bound}

For  $\psi \in H^1(\mathbb{R}^3)$, we define the trial state $\Gamma_\psi$ as in \eqref{def:alphapsi}--\eqref{eq:Gamma_trial}.  
Since we require the  normalization condition
$\tr\gamma_\psi = 1/h$, we have to adjust 
  the $L^2$-norm of $\psi$ accordingly, i.e., 
$$ 1/h = \tr \gamma_\psi = \frac 1h \|\psi\|_2^2 + (1+h^{1/2}) \|\alpha_\psi\|_4^4.$$
Together with \eqref{eq:alpha_4}
this implies
$$\big| \|\psi\|_2^2 - 1\big| \lesssim h^2 \|\psi\|_2 \|\nabla \psi\|_2^3 ,$$
and thus $\|\psi\|_2^2 = 1 + O(h^2)$.

The desired upper bound \eqref{eq:upper_bound} is then an immediate consequence of the following  estimates:
   \begin{subequations}
     \begin{align}
       \begin{split}
         &\tr (-h^2 \Delta +E_b/2)\gamma_\psi + \frac{1}{2}
         \int_{\mathbb{R}^6} V\bigl((x-y)/h)\bigr)
         |\alpha_\psi(x,y)|^2
         \ddd{3} x \ddd{3}y \\
         &\qquad= h \int_{\mathbb{R}^3}\left( \frac{1}{4}
           |\nabla\psi(x)|^2 + \frac 12 g_{\rm bcs}|\psi(x)|^4 \right) \ddd{3}x
         + O(h^{3/2}) \label{eq:upper_bound:H}
       \end{split}
       \\
       &\tr h^2W\gamma_\psi = h\int_{\mathbb{R}^3}
       W(x)|\psi(x)|^2\ddd{3}x + O(h^2) \label{eq:upper_bound:W} \\
       &- \frac{1}{2}\int_{\mathbb{R}^6}|\gamma_\psi(x,y)|^2
       V\Bigl(\frac{x-y}{h}\Bigr)\ddd{3}x \ddd{3}y = \frac h2 g_{\rm ex}\,
       \int_{\mathbb{R}^3} |\psi(x)|^4 \ddd{3}x +
       O(h^{2}) \label{eq:upper_bound:ex}
       \\
       &\int_{\mathbb{R}^6} \gamma_\psi(x,x)\gamma_\psi(y,y)
       V\Bigl(\frac{x-y}{h}\Bigr)\ddd{3}x \ddd{3}y = \frac h2 g_{\rm
         dir}\,\int_{\mathbb{R}^3} |\psi(x)|^4 \ddd{3}x +
       O(h^{2}), \label{eq:upper_bound:dir}
     \end{align}
   \end{subequations}
   where the constants $g_{\rm bcs}$, $g_{\rm ex}$ and $g_{\rm dir}$ are defined in Remark \ref{gis}.  The remainder of this section will be devoted to the proof of these estimates.

 \subsection{Kinetic and potential energy (Proof of
   \eqref{eq:upper_bound:H})}
 \label{sec:upper_bound:H}

Eq.~\eqref{eq:upper_bound:H} is an immediate consequence of the 
 calculation in \eqref{nablaterm} and the bound \eqref{eq:Delta_alpha_psi_4}, using the definition \eqref{eq:Gamma_trial} of $\gamma_\psi$.

 \subsection{External potential (Proof of \eqref{eq:upper_bound:W})}
 \label{sec:upper_bound:W}

 By \eqref{eq:alpha_4} of Corollary~\ref{cor:ineqgamma}, we obtain
  \begin{equation*}
   \tr(h^2 W\alpha_\psi\overline{\alpha_\psi}\alpha_\psi\overline{\alpha_\psi}) \leq  h^2 \|W\|_\infty \tr
   (\alpha_\psi\overline{\alpha_\psi}\alpha_\psi\overline{\alpha_\psi})
   \lesssim h^3.
 \end{equation*}
The leading contribution is thus given by $ h^2 \tr (W\alpha_\psi\overline{\alpha_\psi}) $, which we can write as 
 \begin{equation}
   h^2 \tr (W\alpha_\psi\overline{\alpha_\psi}) = h^2
   \int_{\mathbb{R}^6} W(x)|\alpha_\psi(x,y)|^2 \ddd{3} x\ddd{3} y 
   = h \int_{\mathbb{R}^6} W(X) |\psi(X - h r/2)|^2 |\alpha_0(r)|^2
   \ddd{3} X\ddd{3} r \,.
 \end{equation}
 From  the fundamental theorem of calculus we obtain
 \begin{align*}
   h^2 \tr (W\alpha_\psi\overline{\alpha_\psi}) &= h
   \int_{\mathbb{R}^6} W(X) |\psi(X)|^2
   |\alpha_0(r)|^2 \ddd{3} X\ddd{3} r \\
   &\qquad +h \int_{\mathbb{R}^6} \int_0^1 W(X)
   \frac{\partial}{\partial \tau}|\psi(X - \tau h r/2)|^2
   |\alpha_0(r)|^2 \dd{\tau} |\alpha_0(r)|^2 \ddd{3} X\ddd{3} r.
 \end{align*}
 Using the Cauchy-Schwarz inequality for the integration over the $X$ variable, the last integral is bounded by
 \begin{align*}
   &\left|h \int_{\mathbb{R}^6} \int_0^1 W(X) \Re \big(h\,r\cdot
     \nabla \psi(X - \tau h r/2) \overline{\psi(X - \tau h r/2)}\big)
     \dd{\tau}
     |\alpha_0(r)|^2 \ddd{3} X\ddd{3} r\right|\\
   &\qquad\leq h^2 \|W\|_\infty \|\nabla \psi\|_2 \|\psi\|_2
   \|\sqrt{|\cdot|}\alpha_0\|_2^2.
 \end{align*}
 Since $\alpha_0$ is the ground state of the Schr\"odinger operator $-2\Delta + V$ and hence rapidly decaying,  
 $ \|\sqrt{|\cdot|}\alpha_0\|_2$ is finite.
 This shows \eqref{eq:upper_bound:W}.

 \subsection{Direct and exchange term (Proof of
   \eqref{eq:upper_bound:ex} and \eqref{eq:upper_bound:dir})}
 \label{sec:upper_bound:ex}
 We first argue that the leading order contribution of the direct and
 exchange terms originates from replacing $\gamma_\psi$ by
 $\alpha_\psi\overline{\alpha_\psi}$.  To see this, we simply estimate the
 differences
 \begin{subequations}
   \begin{equation}
     \label{eq:A_ex}
     \int_{\mathbb{R}^6}|\gamma_\psi(x,y)|^2
     V\big((x-y)/h\big)\ddd{3}x \ddd{3}y -
     \int_{\mathbb{R}^6}|(\alpha_\psi\overline{\alpha_\psi})(x,y)|^2
     V\big((x-y)/h\big)\ddd{3}x \ddd{3}y
   \end{equation}
   and
   \begin{equation}
     \label{eq:A_dir}
     \int_{\mathbb{R}^6} \gamma_\psi(x,x)\gamma_\psi(y,y)
     V\big((x-y)/h\big)\ddd{3}x \ddd{3}y -
     \int_{\mathbb{R}^6}
     (\alpha_\psi\overline{\alpha_\psi})(x,x)(\alpha_\psi\overline{\alpha_\psi})(y,y)
     V\big((x-y)/h\big)\ddd{3}x \ddd{3}y.
   \end{equation}
 \end{subequations}
 Both expressions  can be bounded using the following lemma, whose
 proof is elementary.
 \begin{lemma}
   \label{lemma:diff:dir_ex}
   Let $\sigma(x,y)$ and $\delta(x,y)$ be integral kernels of two
   positive trace class operators.  Then
   \begin{subequations}
     \begin{equation}
       \label{eq:diff:dir}
       \begin{split}
         &\left|\int_{\mathbb{R}^6} V(x-y)
           \left[(\sigma+\delta)(x,x)(\sigma+\delta)(y,y) -
             \sigma(x,x)\sigma(y,y) \right] \ddd{3}x \ddd{3}y
         \right|\\
         &\qquad\leq 2 \int_{\mathbb{R}^6} |V(x-y)|
         (\sigma+\delta)(x,x)\delta(y,y)\ddd{3}x \ddd{3}y
       \end{split}
     \end{equation}
     and
     \begin{equation}
       \label{eq:diff:ex}
       \begin{split}
         \left|\int_{\mathbb{R}^6} V(x-y) \left[
             |(\sigma+\delta)(x,y)|^2 - |\sigma(x,y)|^2 \right]
           \ddd{3}x\ddd{3}y \right| \leq 2 \int_{\mathbb{R}^6}
         |V(x-y)| (\sigma+\delta)(x,x)\delta(y,y)\ddd{3}x \ddd{3}y.
       \end{split}
     \end{equation}
   \end{subequations}
 \end{lemma}
 \begin{proof}
   To show \eqref{eq:diff:dir}, we simply use
   \begin{align*}
     &(\sigma+\delta)(x,x)(\sigma+\delta)(y,y) - \sigma(x,x)\sigma(y,y)\\
     &\qquad=
     (\sigma+\delta)(x,x)\delta(y,y) + \delta(x,x)\sigma(y,y)\\
     &\qquad\leq (\sigma+\delta)(x,x)\delta(y,y) +
     \delta(x,x)(\sigma+\delta)(y,y).
   \end{align*}
   Eq.~\eqref{eq:diff:dir} then follows by symmetry, $V(x-y) = V(y-x)$. 

   For \eqref{eq:diff:ex} we follow a similar strategy and first split
   \begin{align*}
     &V(x-y) \left[
       |(\sigma+\delta)(x,y)|^2 - |\sigma(x,y)|^2 \right] \\
     &\qquad= V(x-y) \left[ \overline{(\sigma+\delta)(x,y)}\delta(x,y)
       +
       \overline{\delta(x,y)}\sigma(x,y) \right] \\
     &\qquad\leq |V(x-y)|\left[ |(\sigma+\delta)(x,y)| \,
       |\delta(x,y)| + |\delta(x,y)| \,|\sigma(x,y)| \right].
   \end{align*}
   Applying to $\sigma, \delta$, and $ \sigma+\delta$ the fact that for
   positive trace class operators $a$ its kernel satisfies $$|a(x,y)|
   \leq \sqrt{|a(x,x)|}\sqrt{|a(y,y)|}\,,$$ together with the Cauchy-Schwarz
   inequality, we obtain the stated inequality.
 \end{proof}
 By applying Lemma~\ref{lemma:diff:dir_ex} to $\sigma+\delta =
 \gamma_\psi$ and $\sigma = \overline{\alpha_\psi}\alpha_\psi$ the
 differences \eqref{eq:A_ex} and \eqref{eq:A_dir} can be bounded by
 \begin{equation}
   \label{eq:dir_ex:bound}
   \begin{split}
     & 4(1+h^{1/2}) \int_{\mathbb{R}^6} \left|V\big((x-y)/h\big)\right|
     \gamma_\psi(x,x)
     (\overline{\alpha_\psi}\alpha_\psi\overline{\alpha_\psi}\alpha_\psi)(y,y)
     \ddd{3}x \ddd{3}y \\
     &\qquad\lesssim  \|\gamma_\psi(\cdot,\cdot)\|_\infty
     \int_{\mathbb{R}^6}
     (\overline{\alpha_\psi}\alpha_\psi\overline{\alpha_\psi}\alpha_\psi)(x,x)\left|
       V\big((x-y)/h\big) \right| \ddd{3}x\ddd{3}y
     \\
     &\qquad \lesssim h^3 \|V\|_1
     \|\gamma_\psi(\cdot,\cdot)\|_\infty
     \tr(\overline{\alpha_\psi}\alpha_\psi\overline{\alpha_\psi}\alpha_\psi)
     \lesssim h^2,
   \end{split}
 \end{equation}
 where we used \eqref{eq:gamma_psi_diag} and \eqref{eq:alpha_4} in the last step.

 In order to recover the $\| \psi\|_4^4$ contribution we inspect the
 remaining parts of the direct and the exchange terms separately. We
 begin with the exchange term and write explicitly
 \begin{align*}
   &-\frac{1}{2}
   \int_{\mathbb{R}^6}|(\alpha_\psi\overline{\alpha_\psi})(x,y)|^2
   V\big((x-y)/h\big)\ddd{3}x \ddd{3}y \\
   &\qquad= - \frac{1}{2}\int_{\mathbb{R}^{12}}\alpha_\psi(x,z)
   \overline{\alpha_\psi(z,y)} \alpha_\psi(x,w)
   \overline{\alpha_\psi(w,y)} V\big((x-y)/h\big)\ddd{3}x
   \ddd{3}y\ddd{3}z \ddd{3}w.
 \end{align*}

 Introducing new variables
 \begin{equation*}
   X = \frac{x+y}{2}, \,\,  r = x-y,\,\,   s = x-z, \,\,   t = x-w,
 \end{equation*}
 and rescaling $r/h \to r, s/h \to s, t/h \to t$,
 the last expression becomes
 \begin{align*}
   &  - \frac{h}{2}\int_{\mathbb{R}^{12}} V(r)
   \alpha_0(s)\overline{\alpha_0(r-s)}
   \alpha_0(t)\overline{\alpha_0(r-t)}\\
   &\qquad\quad \times \psi(X+h(r-s)/2) \overline{\psi(X-h s/2)}
   \psi(X-ht/2)\overline{\psi(X+h(r-t)/2)} \ddd{3}X \ddd{3}r\ddd{3}s
   \ddd{3}t\,.
 \end{align*}
The latter equals  
$$
   \frac h2 g_{\rm ex}\,\int_{\mathbb{R}^3} |\psi(x)|^4 \ddd{3}x +
     A_{\rm ex},
 $$
 where
 \begin{align*}
   A_{\rm ex} &= - \frac{h}{2}\int_{\mathbb{R}^{12}}  \ddd{3}X
   \ddd{3}r\ddd{3}s \ddd{3}t \,
   V(r)\alpha_0(s)\overline{\alpha_0(r-s)}
   \overline{\alpha_0(t)}\alpha_0(r-t) \times\\
   &\qquad \times \int_0^1 \frac{\dd}{\dd \tau} \left(\psi(X+\tau
     h(r-s)/2) \overline{\psi(X-\tau h s/2)} \psi(X-\tau h
     t/2)\overline{\psi(X+\tau h(r-t)/2)}\right) \dd{\tau}\,.
 \end{align*}
 This can be bounded by
 \begin{equation*}
   |A_{\rm ex} | \lesssim  h^2 \|\nabla \psi\|_2\|\psi\|_6^3 \|V (\alpha_0*\alpha_0)((|\cdot| \alpha_0)*\alpha_0)\|_1
   \lesssim
    h^2 \|\nabla \psi\|_2^4 \|V\|_1 \|\alpha_0\|_2^3 \big\||\cdot|\alpha_0\big\|_2,
 \end{equation*}
 using the H\"older, Sobolev  and Cauchy-Schwarz inequalities.  This shows \eqref{eq:upper_bound:ex}.

  We
 continue with the direct term.  Its remaining part is given by
 \begin{align*}
   &\int_{\mathbb{R}^6}
   (\alpha_\psi\overline{\alpha_\psi})(x,x)(\alpha_\psi\overline{\alpha_\psi})(y,y)
   V\big(({x-y})/{h}\big)\ddd{3}x \ddd{3}y \\
   &\qquad= \int_{\mathbb{R}^6}
   |\alpha_\psi(x,z)|^2|\alpha_\psi(y,w)|^2
   V\big((x-y)/h\big)\ddd{3}x \ddd{3}y\ddd{3}w \ddd{3}z \\
   &\qquad= h\int_{\mathbb{R}^{12}} V(r)
   |\alpha_0(s)|^2|\alpha_0(t)|^2 |\psi(X+h(r-s)/2)|^2
   |\psi(X-h(r+t)/2)|^2 \ddd{3}X \ddd{3}r \ddd{3}s \ddd{3}t,
 \end{align*}
 where we changed to the variables
 \begin{equation*}
   X = \frac{x+y}{2},\,\,   r = x-y,\,\,   s = x-z, \,\,  t = y-w,
 \end{equation*}
 and rescaled $r, s, t$.
By proceeding as above, we see that this expression equals
 \begin{equation}
   \label{eq:dir:reduced}
    \frac h2 g_{\rm dir}\, \int_{\mathbb{R}^3} |\psi(x)|^4 \ddd{3}x
     +A_{\rm dir},
 \end{equation}
 where
 \begin{align*}
   A_{\rm dir} &= h\int_{\mathbb{R}^{12}} V(r)
   |\alpha_0(s)|^2|\alpha_0(t)|^2\times\\
   &\qquad \times \int_0^1 \frac{\dd}{\dd \tau} \left( |\psi(X+\tau
     h(r-s)/2)|^2 |\psi(X-\tau h(r+t)/2)|^2 \right) \dd{\tau} \ddd{3}X
   \ddd{3}r \ddd{3}s \ddd{3}t
 \end{align*}
 is bounded by
 \begin{align*}
   |A_{\rm dir}| &\lesssim h^2 \|\nabla \psi\|_2 \|\psi\|_6^3\,
   \|\alpha_0\|_2\, \big(\||\cdot| V\|_1 \|\alpha_0\|_2 + \|V\|_1
   \|\sqrt{|\cdot| } \alpha_0\|_2\big).
 \end{align*}
 This shows \eqref{eq:upper_bound:dir}, and thus concludes the proof of the upper bound. 
 \section{Lower bound}
 \label{sec:lower_bound}

 Our proof of the lower bound on $E^{\rm BHF}(h)$ in
 Theorem~\ref{thm:GP} consists of two parts.  As a first step we obtain
 a priori bounds on approximate ground states.

 \begin{lemma}[A priori bounds]
   \label{lemma:apriori}
   Let $\Gamma$ be a state satisfying  $\tr \gamma = 1/h$ and 
   \begin{equation*}
     \mathcal{E}^{\rm
       BHF}(\Gamma) \leq -E_b\frac{1}{2h} + Ch
   \end{equation*}
   for some $C>0$. Define the function $\Psi$ as
    \begin{equation}\label{dpsi}
     \psi(X) = \frac{1}{h}\int_{\mathbb{R}^3}
     \alpha_0(r/h)\alpha(X+ r/2,X- r/2) \ddd{3}r\,,
   \end{equation}
  and define $\widetilde{\xi}(X,r) = \xi(X+r/2, X-r/2)$ through the decomposition 
  $$\widetilde{\alpha}(X,r) := \alpha(X+r/2,X-r/2)=   h^{-2} \psi(X)\alpha_0(r/h) +
  \widetilde{\xi}(X,r)\,.$$
  Then these functions satisfy the bounds
  \begin{subequations}
    \begin{center}
      \begin{minipage}{.45\linewidth}
        \begin{align}
          &\tr \big[ (-h^2\Delta + \tfrac 12 E_b) (\gamma -
          \alpha\overline{\alpha})\big]
          \lesssim h,\label{eq:apriori:Delta_gamma_2}\\
          &\tr (\gamma^2) \leq \tr(\gamma - \alpha\overline{\alpha})
          \lesssim h, \label{eq:apriori:gamma_2}\\
          &\|\psi\|_2 \leq 1,\label{eq:apriori:psi} \\
          &\|\nabla \psi\|_2 \lesssim 1,\label{eq:apriori:nabla_psi}
        \end{align}
      \end{minipage}
      \begin{minipage}{.45\linewidth}
        \begin{align}
          &\|\widetilde{\xi}\|_2 \lesssim h^{1/2},\label{eq:apriori:xi}\\
          \label{eq:apriori:nabla_xi}
          &\|\nabla_X\widetilde{\xi}\|_2 \lesssim h^{-1/2}, \\
          \label{eq:apriori:nabla_r_xi}
          &\|\nabla_r \widetilde{\xi}\|_2 \lesssim h^{-1/2}, \\
          \label{eq:apriori:alpha_4}
          &\tr (\alpha\overline{\alpha}\alpha\overline{\alpha})
            \lesssim h.
        \end{align}
      \end{minipage}
    \end{center}
  \end{subequations}
\end{lemma}
Note that our definition implies that  $\widetilde \xi(X,\,\cdot\,)$ is orthogonal to $  \alpha_0(\,\cdot\,/h)$ for almost all $X$. The norms in \eqref{eq:apriori:xi}--\eqref{eq:apriori:nabla_r_xi} are in $L^2(\mathbb{R}^6)$.
\begin{proof}
  We have seen in Section~\ref{sec:stability} that the sum of the direct and
  exchange terms is non-negative. Consequently,
  \begin{align*}
     h &\gtrsim \mathcal{E}^{\rm BHF}(\Gamma) + E_b  \frac{1}{2h}     \\
    &\geq \tr (-h^2 \Delta +E_b/2)\gamma + \frac{1}{2}
    \int_{\mathbb{R}^6} V\big((x-y)/h\big) |\alpha(x,y)|^2 \ddd{3} x
    \ddd{3}y -h^2\|W\|_\infty \tr(\gamma).
  \end{align*}
  We bring the term $h^2\|W\|_\infty \tr(\gamma) \lesssim h $ to the left
   side.  Adding and subtracting the  expression $\tr (-h^2 \Delta
    +E_b/2)\alpha\overline{\alpha}$ we obtain
  \begin{equation}\label{weo} 
     h \gtrsim \tr (-h^2 \Delta +E_b/2)
    (\gamma-\alpha\overline{\alpha}) + \tr (-h^2 \Delta
    +E_b/2)\alpha\overline{\alpha} + \frac{1}{2} \int_{\mathbb{R}^6}
    V\big((x-y)/h\big) |\alpha(x,y)|^2 \ddd{3} x \ddd{3}y.
  \end{equation}
  The last two terms on the right side can be
  expressed via center-of-mass and relative coordinates as
  \begin{equation}
    \begin{split}
      \label{calxi}
      &\int_{\mathbb{R}^3} \left \langle \alpha(\cdot, y),
        \left[-h^2\Delta_x + \frac{1}{2}V\Bigl(\frac{\cdot
            -y}{h}\Bigr) +\frac{E_b}{2}\right] \alpha(\cdot ,y)\right
      \rangle_{L^2(\mathbb{R}^3)} \ddd{3}y
      \\
      &\qquad= \left \langle \alpha ,
        \left[-\frac{h^2}{4}\Delta_X -h^2 \Delta_r +
          \frac{1}{2}V(r/h) +\frac{E_b}{2} \right] \alpha\right
      \rangle_{L^2(\mathbb{R}^6)} \\
      &\qquad= \frac{h}{4} \int_{\mathbb{R}^3} |\nabla \psi(X)|^2
      \ddd{3} X + \frac{h^2}4  \|\nabla_X \widetilde{\xi}\|_2^2 \\ & \quad\qquad  + 
      \int_{\mathbb{R}^3} \left \langle \widetilde{\xi}(X,\cdot),
      \left[ -h^2\Delta_r + \frac{1}{2} V(\cdot/h) +
      E_b/2 \right] \widetilde{\xi}(X,\cdot)\right\rangle_{L^2(\mathbb{R}^3)}  \ddd{3} X,
    \end{split}
  \end{equation}
  where we used that $\alpha_0$ is the normalized  zero energy eigenvector
of the operator $-\Delta + V/2 + E_b/2$, as well as  the fact that
  $\widetilde{\xi}(X,\cdot)$ is orthogonal to $\alpha_0(\cdot/h)$ for almost
  every $X \in\mathbb{R}^3$.  Hence \eqref{weo} implies
  \begin{multline} \label{combw}
    h \gtrsim
     \tr (-h^2 \Delta +E_b/2)
    (\gamma-\alpha\overline{\alpha})     + \frac{h}{4} \|\nabla \psi\|_2^2
    + \frac{h^2}4 \|\nabla_X \widetilde{\xi}\|_2^2\\
    + \int_{\mathbb{R}^3} \langle
    \widetilde{\xi}(X,\cdot), (-h^2\Delta + \frac{1}{2} V(\cdot/h)+
    E_b/2)\widetilde{\xi}(X,\cdot)\rangle \ddd{3} X.
  \end{multline}
  Since all terms on the right  side are non-negative, and $\gamma -
  \gamma^2 \geq \alpha \bar \alpha $, we immediately obtain the estimates
  \eqref{eq:apriori:Delta_gamma_2}, \eqref{eq:apriori:gamma_2}, \eqref{eq:apriori:nabla_psi} and 
  \eqref{eq:apriori:nabla_xi}.

According to
  Assumption \ref{asm:V} the operator $-\Delta + V/2$ has a spectral gap
   between the ground state energy $-E_b/2$ and the next 
  eigenvalue.  This implies that we can find a $\kappa > 0$ and an $\varepsilon > 0$
  such that
$$-(1-\varepsilon)\Delta + V/2 + E_b/2  \geq \kappa$$
on the orthogonal complement of $\alpha_0$. Hence
\begin{equation*}
  \int_{\mathbb{R}^3} \langle \widetilde{\xi}(X,\cdot),(-h^2
  \Delta_r + \frac{1}{2} V(\cdot/h) + E_b/2)
  \widetilde{\xi}(X,\cdot) \rangle \ddd{3}X
  \geq \kappa \| \widetilde{\xi} \|_2^2
  + \varepsilon h^2 \| \nabla_r\widetilde{\xi} \|_2^2\,.
\end{equation*}
In combination with \eqref{combw} this yields the estimates \eqref{eq:apriori:xi} and \eqref{eq:apriori:nabla_r_xi}.  

Since $\|{\alpha}\|^2_2 \leq
 \tr \gamma = 1/h$ we obtain for the $L^2$-norm of $\psi$, that, by definition \eqref{dpsi} and the Cauchy-Schwarz inequality, 
\begin{equation}
  \label{eq:apriori:psi:derivation}
  \begin{split}
    \|\psi\|_2^2 &= h^{-2} \int_{\mathbb{R}^3} \alpha_0
    (r_1/h)\widetilde{\alpha}(X,r_2)\alpha_0({r_2}/{h})
    \overline{\widetilde{\alpha}(X,r_1)}
    \ddd{3}r_1\ddd{3}r_2\ddd{3}X \\
    &\leq h^{-2} \int_{\mathbb{R}^3} |\alpha_0({r_1}/{h})|^2
    |\widetilde{\alpha}(X,r_2)|^2 \ddd{3}r_1\ddd{3}r_2 \ddd{3}X = h
    \|\alpha\|_2^2 \leq 1,
  \end{split}
\end{equation}
implying \eqref{eq:apriori:psi}.  Finally, to see \eqref{eq:apriori:alpha_4} note that since $\gamma\geq \alpha\overline{\alpha}$
\begin{equation*}
  \tr(\alpha\overline{\alpha}\alpha\overline{\alpha}) =
  \tr\big(\gamma^2 - \gamma(\gamma-\alpha\overline{\alpha}) -
  (\gamma-\alpha\overline{\alpha})\gamma +
  (\gamma-\alpha\overline{\alpha})^2\big) 
  \leq \tr(\gamma^2) + \tr (\gamma-\alpha\overline{\alpha})^2  \lesssim h.
\end{equation*}
\end{proof}
Observe that we do not necessarily have $\|\psi\|_2^2 = 1$.  The norm
deviates from $1$ by a correction of order $h^2$,
\begin{equation}
  \label{eq:psi_constraint}
1 -  \|\psi\|_2^2 
  = h\left(   \tr \gamma -  \tr \alpha_{\psi}\overline{\alpha_{\psi}} \right)
  \leq  h \tr(\gamma -\alpha\overline{\alpha}) + h \left|
    \tr(\alpha_{\psi}\overline{\alpha_{\psi}} - \alpha\overline{\alpha}) \right| \,.\end{equation}
By \eqref{eq:apriori:gamma_2} and \eqref{eq:apriori:xi} (and the orthogonality of $\widetilde \xi$ and $\alpha_0$), the right
 side is $O(h^2)$.

With the aid of the function $\psi$  we can define a
corresponding state $\Gamma_\psi$ as in \eqref{def:alphapsi}--\eqref{eq:Gamma_trial}.  By multiplying $\psi$ with a factor $\lambda = 1 +O(h^2)$ we can assume 
that $\tr \gamma_{\lambda\psi} = 1/h$. 
The second step now consists of proving that for a lower bound we can
replace $\mathcal{E}^{\rm BHF}(\Gamma)$ by $\mathcal{E}^{\rm
  BHF}(\Gamma_{\lambda\psi})$ up to higher order terms. Together with the
calculations from the upper bound this implies the lower bound stated
in Theorem \ref{thm:GP}.
\begin{lemma}\label{lem4}
 With $\Gamma$ and $\Gamma_{\lambda\psi}$ defined as above, one
  has
  \begin{equation}
    \label{lowerbound}
    \mathcal{E}^{\rm
      BHF}(\Gamma) \geq \mathcal{E}^{\rm
      BHF}(\Gamma_{\lambda\psi})  -  O(h^{3/2}).
  \end{equation}
\end{lemma}
Lemma~\ref{lem4} not only completes the proof of the lower bound \eqref{eq:lower_bound}, it also allows to establish the claim about approximate minimizers in Eqs.~\eqref{apprmin}--\eqref{eq:approx_minimizer} in Theorem~\ref{thm:GP}. Given an approximate minimizer satisfying \eqref{apprmin}, Lemma~\ref{lemma:apriori} yields \eqref{eq:alpha_psi:0}, while a combination of \eqref{lowerbound} and \eqref{eq:BHFequGP} implies \eqref{eq:approx_minimizer}.

It remains to prove the bound \eqref{lowerbound}, which is an immediate consequence of the  following estimates: 
 \begin{subequations}
\begin{align}
\begin{split}
        &\tr (-h^2 \Delta +E_b/2)\gamma + \frac{1}{2} \int_{\mathbb{R}^6}
        V\big((x-y)/h\big)
        |\alpha(x,y)|^2 \ddd{3} x \ddd{3}y\\
        &\qquad \qquad \geq (-h^2 \Delta +E_b/2)\gamma_{\lambda\psi} + \frac{1}{2} \int_{\mathbb{R}^6}
        V\big((x-y)/h\big) |\alpha_{\lambda\psi}(x,y)|^2 \ddd{3} x \ddd{3}y -
        O(h^{3/2})\label{eq:lower_bound:H}
\end{split}
\end{align}
\begin{equation} \label{eq:lower_bound:W} 
h^2\tr W\gamma \geq h^2 \tr W\gamma_{\lambda\psi} - O(h^{2})
\end{equation}
\begin{equation}
- \int_{\mathbb{R}^6}|\gamma(x,y)|^2
      V\big((x-y)/h\big)\ddd{3}x \ddd{3}y \geq -
      \int_{\mathbb{R}^6}|\gamma_{\lambda\psi}(x,y)|^2
      V\big((x-y)/h\big)\ddd{3}x \ddd{3}y - O(h^{3/2})
      \label{eq:lower_bound:ex}
\end{equation}
\begin{equation}
\int_{\mathbb{R}^6} \gamma(x,x)\gamma(y,y)
      V\big((x-y)/h\big)\ddd{3}x \ddd{3}y \geq \int_{\mathbb{R}^6}
      \gamma_{\lambda\psi}(x,x)\gamma_{\lambda\psi}(y,y) V\big((x-y)/h\big)\ddd{3}x
      \ddd{3}y - O(h^{3/2}).\label{eq:lower_bound:dir}
\end{equation}
  \end{subequations}
  The remainder of this section will be dedicated to proving these estimates.

\subsection{Kinetic and potential energy (Proof of
  \eqref{eq:lower_bound:H})}
\label{sec:lower_bound:H}
Let us  decompose $\gamma$ according to
  \begin{equation*}
    \gamma = \alpha\overline{\alpha} +
    \alpha\overline{\alpha}\alpha\overline{\alpha}
    + (\gamma - \alpha\overline{\alpha} - \gamma^2) +
    (\gamma-\alpha\overline{\alpha})^2
    + \alpha\overline{\alpha}(\gamma-\alpha\overline{\alpha})
    + (\gamma-\alpha\overline{\alpha})\alpha\overline{\alpha},
  \end{equation*}
  where $(\gamma - \alpha\overline{\alpha} - \gamma^2)$
  and $(\gamma-\alpha\overline{\alpha})^2$ are positive self-adjoint
  operators and thus
  \begin{equation*}
    \tr (-h^2\Delta + E_b/2)\big((\gamma -
    \alpha\overline{\alpha} - \gamma^2)
    + (\gamma-\alpha\overline{\alpha})^2\big) \geq 0.
  \end{equation*}
  Adding and subtracting the term $$ \tr (-h^2 \Delta + E_b/2 )\gamma_{\lambda\psi}
      + \frac{1}{2} \int_{\mathbb{R}^6}
    V\big((x-y)/h\big) |\alpha_{\lambda\psi}(x,y)|^2 \ddd{3} x \ddd{3}y,$$
   we obtain
  \begin{align}\nonumber 
    & \tr (-h^2 \Delta +  E_b/2)\gamma  + \frac{1}{2} \int_{\mathbb{R}^6}
    V\big((x-y)/h\big)
    |\alpha(x,y)|^2 \ddd{3} x \ddd{3}y \\ \nonumber 
 &  \geq 
     \tr (-h^2 \Delta +E_b/2)\gamma_{\lambda\psi}
      + \frac{1}{2} \int_{\mathbb{R}^6}
    V\big((x-y)/h\big) |\alpha_{\lambda\psi}(x,y)|^2 \ddd{3} x \ddd{3}y\\ \nonumber 
    & \quad + 
      \tr \big((-h^2\Delta +E_b/2)\alpha\overline{\alpha}\big) +
      \frac{1}{2} \int_{\mathbb{R}^6} V\big((x-y)/h\big)
      |\alpha(x,y)|^2 \ddd{3} x \ddd{3}y\\ \nonumber
      &\quad- \tr \big((-h^2\Delta +E_b/2)\alpha_\lpsi
      \overline{\alpha_\lpsi}\big) - \frac{1}{2} \int_{\mathbb{R}^6}
      V\big((x-y)/h\big) |\alpha_\lpsi(x,y)|^2 \ddd{3} x \ddd{3}y \\ \nonumber
    &\quad + \tr \big [(-h^2 \Delta + E_b/2)
    \alpha\overline{\alpha}\alpha\overline{\alpha} \big]
    -\tr \big [(-h^2 \Delta + E_b/2)
      \alpha_\lpsi\overline{\alpha_\lpsi}\alpha_\lpsi\overline{\alpha_\lpsi}
    \big]\\
    &\quad + \tr \big [(-h^2 \Delta + E_b/2)\alpha\overline{\alpha}(\gamma-\alpha\overline{\alpha})
    + (\gamma-\alpha\overline{\alpha})\alpha\overline{\alpha}\big]
      -h^{1/2} \tr \big [(-h^2 \Delta + E_b/2)
      \alpha_\lpsi\overline{\alpha_\lpsi}\alpha_\lpsi\overline{\alpha_\lpsi}\big]. \label{rhss}
  \end{align}
The identity \eqref{calxi} immediately implies that  
  \begin{multline}
    \tr \big((-h^2\Delta +E_b/2)\alpha\overline{\alpha}\big) +
      \frac{1}{2} \int_{\mathbb{R}^6} V\big((x-y)/h\big)
      |\alpha(x,y)|^2 \ddd{3} x \ddd{3}y\\
      \geq  \tr \big((-h^2\Delta +E_b/2)\alpha_\psi
      \overline{\alpha_\psi}\big) + \frac{1}{2} \int_{\mathbb{R}^6}
      V\big((x-y)/h\big) |\alpha_\psi(x,y)|^2 \ddd{3} x \ddd{3}y = \frac h 4 \|\nabla\psi\|_2^2 \,,
  \end{multline}
and hence the sum of the  second and third lines on the right side of \eqref{rhss} is bounded from below by $\frac h 4 (1-\lambda^2) \|\nabla\psi\|_2^2 = O(h^3)$.   
  Hence the proof of \eqref{eq:lower_bound:H} reduces to establishing the  estimates
  \begin{align}
    \label{eq:diff:alpha_4}
    \left|\tr (-h^2 \Delta +
    E_b/2)\big[\alpha\overline{\alpha}\alpha\overline{\alpha}
    -
    \alpha_\lpsi\overline{\alpha_\lpsi}\alpha_\lpsi\overline{\alpha_\lpsi}\big] \right|
    &\lesssim  h^{2}, \\
    \label{eq:lower_bound:H:R1}
   | \tr\big((-h^2 \Delta +
    E_b/2)\big[\alpha\overline{\alpha}(\gamma-\alpha\overline{\alpha})
    + (\gamma-\alpha\overline{\alpha})\alpha\overline{\alpha}\big]\big)|
    & \lesssim h^{3/2},\\
    \label{eq:lower_bound:H:R2}
    h^{1/2}\tr\left((-h^2 \Delta + E_b/2)
      \alpha_\lpsi\overline{\alpha_\lpsi}\alpha_\lpsi\overline{\alpha_\lpsi}\right)
    & \lesssim h^{3/2},
  \end{align}
  which we are going to show in the following.
 
 Inequality~\eqref{eq:lower_bound:H:R2}
 is an immediate consequence of \eqref{eq:Delta_alpha_psi_4}. 
It also implies that it is enough to show  \eqref{eq:diff:alpha_4} for $\lambda =1$. To do this, we rewrite
\begin{equation}
 \alpha\overline{\alpha}\alpha\overline{\alpha}
      -\alpha_\psi\overline{\alpha_\psi}\alpha_\psi\overline{\alpha_\psi}
      =
      \alpha_\psi\overline{\alpha}\alpha\overline{\xi}
      +\xi\overline{\alpha}\alpha\overline{\alpha_\psi}
      + \xi\overline{\alpha}\alpha\overline{\xi}
      +\alpha_\psi(\overline{\alpha}\alpha-\overline{\alpha_\psi}\alpha_\psi)\overline{\alpha_\psi}\,.
\end{equation}
With $H := -h^2\Delta + \frac{E_b}{2}$ we obtain with the aid of 
H\"older's inequality for traces 
  \begin{equation}
    \label{eq:alpha_4:E}
    \begin{split}
      |\tr\big( H\big [
      \alpha\overline{\alpha}\alpha\overline{\alpha}
      -\alpha_\psi\overline{\alpha_\psi}\alpha_\psi\overline{\alpha_\psi}
      \big]\big)|
      &=
      \left|\tr\big(H^{1/2} \big[
      \alpha_\psi\overline{\alpha}\alpha\overline{\xi}
      +\xi\overline{\alpha}\alpha\overline{\alpha_\psi}
      +\xi\overline{\alpha}\alpha\overline{\xi}
      +\alpha_\psi(
      \overline{\alpha}\alpha-\overline{\alpha_\psi}\alpha_\psi)
      \overline{\alpha_\psi}
      \big]H^{1/2} \big) \right|\\
      &\leq
      2\|H^{1/2} \alpha_\psi\|_6 \|\alpha\|_6^2 \|H^{1/2}\xi\|_2
      +\|\alpha\|_\infty^2\|H^{1/2}\xi\|_2^2\\
      &\qquad+ \|H^{1/2} \alpha_\psi\|_6^2
      \|\overline{\alpha}\alpha -\overline{\alpha_\psi}\alpha_\psi\|_{3/2}.
    \end{split}
  \end{equation}
  Note that for any operator $T$, we have
  \begin{align*}
    \|H^{1/2} T\|_{2n} &= \|T^*HT\|_n^{1/2}
    \leq \sqrt{\|T^*(-h^2 \Delta)T\|_n + \tfrac 12 E_b \|T^*T\|_n}
    \leq h \|\nabla T\|_{2n} + \sqrt{\frac{E_b}{2}} \| T\|_{2n}\\
    &\leq h \left(\frac{1}{2}\|\nabla_X T\|_{2n} + \|\nabla_r
    T\|_{2n}\right) + \sqrt{\frac{E_b}{2}} \| T\|_{2n}\,,
  \end{align*}
  where in the last  line the operators $\nabla_X T$ and $\nabla_r T$ are defined via the kernels $(\nabla_X T)(x,y)$ and $(\nabla_r T)(x,y)$, respectively. 
  By applying this to the terms in  \eqref{eq:alpha_4:E} we obtain
  \begin{align}
    \label{eq:H_alpha_6}
    \|H^{1/2} \alpha_\psi\|_6
    &\leq \frac{h}{2} \|\nabla_X \alpha_\psi\|_6
      +h \|\nabla_r \alpha_\psi\|_6
      +\sqrt{\frac{E_b}{2}} \| \alpha_\psi\|_6 \lesssim h^{1/2},\\
    \label{eq:H_xi_2}
    \|H^{1/2} \xi\|_2 &\leq \frac{h}{2} \|\nabla_X \xi\|_2
      +h \|\nabla_r \xi\|_2
      +\sqrt{\frac{E_b}{2}} \| \xi\|_2 \lesssim h^{1/2},
  \end{align}
  where we used
  $\|\nabla_X \alpha_\psi\|_6 = \|\alpha_{\nabla \psi}\|_6 \leq
  \|\alpha_{\nabla \psi}\|_2 = \|\nabla \psi\|_2h^{-1/2}$, together with
  \eqref{eq:alpha_6}, \eqref{eq:nabla_alpha_6},  \eqref{eq:apriori:xi},
  \eqref{eq:apriori:nabla_xi} and \eqref{eq:apriori:nabla_r_xi}.
  The term
  $\|\overline{\alpha}\alpha-\overline{\alpha_\psi}\alpha_\psi\|_{3/2}$
  in \eqref{eq:alpha_4:E} can be bounded by
  \begin{equation}
    \label{eq:d_alpha_2}
    \|\overline{\alpha}\alpha-\overline{\alpha_\psi}\alpha_\psi\|_{3/2}
    =\|\overline{\alpha_\psi}\xi + \overline{\xi}\alpha_\psi +
    \overline{\xi}\xi \|_{3/2}
    \leq
    2 \|\alpha_\psi\|_6\|\xi\|_2 + \|\xi\|_6\|\xi\|_2
    \leq
    2 \|\alpha_\psi\|_6\|\xi\|_2 + \|\xi\|_2^2
    \lesssim h,
  \end{equation}
 where we used \eqref{eq:alpha_6} and \eqref{eq:apriori:xi}.
By combining \eqref{eq:alpha_4:E} with \eqref{eq:H_alpha_6}--\eqref{eq:d_alpha_2} we obtain 
  \eqref{eq:diff:alpha_4}.

  To show \eqref{eq:lower_bound:H:R1}, we can bound 
  \begin{align*}
  \left|   \tr\left(H  [ (\gamma - \alpha\overline{\alpha})
    \alpha\overline{\alpha} +  \alpha\overline{\alpha}  (\gamma - \alpha\overline{\alpha}) ]  \right) \right| 
    &= 2 \left| \Re \tr\left(H^{1/2} (\gamma - \alpha\overline{\alpha})H^{1/2}
      \frac{1}{H^{1/2}}\alpha\overline{\alpha}H^{1/2}\right) \right| \\
    &\leq  2 \tr   H(\gamma - \alpha\overline{\alpha})
      \left\| \frac{1}{H^{1/2}} \alpha\overline{\alpha}H^{1/2}\right\|_\infty.
  \end{align*}
The first factor on the right side is bounded by $O(h)$ according to \eqref{eq:apriori:Delta_gamma_2}.   Moreover, 
  \begin{equation*}
    \left\| \frac{1}{H^{1/2}}
      \alpha\overline{\alpha}H^{1/2}\right\|_\infty
    \leq \sqrt\frac 2{E_b} \| \alpha\overline{\alpha} H \alpha\overline{\alpha}\|_\infty^{1/2} \leq \sqrt\frac 2{E_b} \left( \tr H \alpha\overline{\alpha}\alpha\overline{\alpha} \right)^{1/2} \,,
  \end{equation*}
 which is bounded by $O(h^{1/2})$ using \eqref{eq:diff:alpha_4} together with
  \eqref{eq:Delta_alpha_psi_4}. This proves 
   \eqref{eq:lower_bound:H:R1}.

  \subsection{External potential (Proof of \eqref{eq:lower_bound:W})}
  \label{sec:lower_bound:W}
 Since $W$ is bounded, $\tr\gamma_\psi = O(h^{-1})$ and $\lambda = 1 + O(h^2)$, it clearly suffices to consider the case $\lambda = 1$.  Using the form \eqref{eq:Gamma_trial}  of $\gamma_\psi$ we evaluate
  \begin{align}\nonumber 
    h^2 \tr W(\gamma - \gamma_\psi) & =
    h^2 \tr W(\gamma-\alpha\overline{\alpha}) + h^2 \tr W(\alpha\overline{\alpha} - \alpha_\psi\overline{\alpha_\psi}) - (1 + h^{1/2} ) h^2 \tr W(\alpha_\psi\overline{\alpha_\psi} \alpha_\psi\overline{\alpha_\psi}) \\
  & \nonumber  \geq - h^2 \|W\|_\infty \left[ \tr(
      \gamma-\alpha\overline{\alpha}) + \|\xi\|_2^2 + 2 \| \alpha_\psi
      \overline{\xi} \|_1 + (1+ h^{1/2}) \tr \alpha_\psi\overline{\alpha_\psi}
      \alpha_\psi\overline{\alpha_\psi})\right]  \\ & \geq - O(h^{2}),
  \end{align}
  where we used \eqref{eq:apriori:gamma_2}, the decomposition $\alpha
  = \alpha_\psi + \xi$, and $\| \alpha_\psi \xi \|_1 \leq
  \|\alpha_\psi\|_2 \|\xi\|_2 \lesssim 1$.

    \subsection{Direct and exchange term (Proof of
    \eqref{eq:lower_bound:ex} and \eqref{eq:lower_bound:dir})}
  \label{sec:lower_bound:ex}

  Our strategy is as follows.
  As  a first step we reduce the direct term and exchange term
  to  corresponding expressions involving $\alpha$ only, and show that 
   \begin{align}
    \label{eq:lower_bound:ex:1}
    \left|\int_{\mathbb{R}^6}|\gamma(x,y)|^2
    V\big((x-y)/h\big)\ddd{3}x \ddd{3}y
    -\int_{\mathbb{R}^6}|(\alpha\overline{\alpha})(x,y)|^2
    V\big((x-y)/h\big)\ddd{3}x \ddd{3}y\right| & \lesssim h^{2}, 
    \end{align}
    \begin{align}
    \label{eq:lower_bound:dir:1}
    \left|\int_{\mathbb{R}^6} \gamma(x,x)\gamma(y,y)
    V\big((x-y)/h\big)\ddd{3}x \ddd{3}y
    - \int_{\mathbb{R}^6}
    (\alpha\overline{\alpha})(x,x)(\alpha\overline{\alpha})(y,y)
    V\big((x-y)/h\big)\ddd{3}x \ddd{3}y\right| & \lesssim h^{2}.
  \end{align}
  As a second step, we show that up to an error $O(h^{3/2})$ we are able to replace $\alpha$  by $\alpha_\psi$
  in the corresponding expressions, i.e.,
  \begin{align}
    \label{eq:lower_bound:ex:2}
    \left|\int_{\mathbb{R}^6}|(\alpha\overline{\alpha})(x,y)|^2
    V\big((x-y)/h\big)\ddd{3}x \ddd{3}y
    -\int_{\mathbb{R}^6}|(\alpha_\psi\overline{\alpha_\psi})(x,y)|^2
    V\big((x-y)/h\big)\ddd{3}x \ddd{3}y\right| & \lesssim h^{3/2}
    \end{align}
    and
   \begin{multline} \biggl|\int_{\mathbb{R}^6} (\alpha\overline{\alpha})(x,x)(\alpha\overline{\alpha})(y,y)
    V\big((x-y)/h\big)\ddd{3}x \ddd{3}y\\
    \label{eq:lower_bound:dir:2}
    - \int_{\mathbb{R}^6}
    (\alpha_\psi\overline{\alpha_\psi})(x,x)(\alpha_\psi\overline{\alpha_\psi})(y,y)
    V\big((x-y)/h\big)\ddd{3}x \ddd{3}y\biggr|  \lesssim h^{3/2}.
\end{multline}
These two steps together, in combination with $\lambda=1 + O(h^2)$, lead to  \eqref{eq:lower_bound:ex} and
  \eqref{eq:lower_bound:dir}, respectively.

  The estimates \eqref{eq:lower_bound:ex:1} and
  \eqref{eq:lower_bound:dir:1} can be obtained by applying
  Lemma~\ref{lemma:diff:dir_ex} with $\sigma = \alpha\overline{\alpha}$ and $\delta = \gamma
  -\alpha\overline{\alpha}$. 
  As a result we obtain that the left  sides of both \eqref{eq:lower_bound:ex:1} and
  \eqref{eq:lower_bound:dir:1} are bounded by
  \begin{subequations}
    \begin{align}
      \nonumber &2\int_{\mathbb{R}^6} \left|V\big((x-y)/h\big) \right| 
        (\gamma-\alpha\overline{\alpha})(x,x) \gamma(y,y)
      \ddd{3}x \ddd{3}y \\
      \label{eq:A:dir_ex:1}
      &\qquad = 2\int_{\mathbb{R}^6} \left|V\big((x-y)/h\big)\right| 
        (\gamma-\alpha\overline{\alpha})(x,x)
        (\gamma-\alpha\overline{\alpha})(y,y)
      \ddd{3}x \ddd{3}y \\
      &\qquad\phantom{\leq}
      \label{eq:A:dir_ex:2}
      + 2\int_{\mathbb{R}^6} \left|V\big((x-y)/h\big)\right| 
        (\gamma-\alpha\overline{\alpha})(x,x)
        (\alpha\overline{\alpha})(y,y) \ddd{3}x \ddd{3}y.
    \end{align}
  \end{subequations}
  By \eqref{eq:apriori:gamma_2}, the term \eqref{eq:A:dir_ex:1} is
  bounded by
  \begin{equation*}
   \int_{\mathbb{R}^6} (\gamma- \alpha\overline{\alpha})(x,x)(\gamma-\alpha\overline{\alpha})(y,y)
      \left| V\big((x-y)/h\big)\right|\ddd{3}x \ddd{3}y
    \leq [\tr(\gamma- \alpha\overline{\alpha})]^2\|V\|_\infty \lesssim h^2.
  \end{equation*}
  For \eqref{eq:A:dir_ex:2}, we are going to use the decomposition
  $\alpha = \alpha_\psi +\xi$ in the form 
  $$ \alpha \bar\alpha = \alpha_\psi\overline{\alpha_\psi}+
  \xi\overline{\alpha_\psi} +\alpha_\psi\overline{\xi} +
  \xi\overline{\xi}$$ and we thus have to bound four terms.  First, 
  observe that
  \begin{align*}
   \int_{\mathbb{R}^6}
      (\gamma-\alpha\overline{\alpha})(x,x)(\xi\overline{\xi})(y,y)
      \left|  V\big((x-y)/h\big)\right|\ddd{3}x \ddd{3}y \leq \|V\|_\infty
    \tr(\gamma-\alpha\overline{\alpha}) \tr(\xi\overline{\xi}) \lesssim
    h^2.
  \end{align*}
  Second,  using \eqref{eq:sup_alpha2}, 
  \begin{multline}
    \int_{\mathbb{R}^6}
      (\gamma-\alpha\overline{\alpha})(x,x)(\alpha_\psi\overline{\alpha_\psi})(y,y)
    \left|  V\big((x-y)/h\big)\right|\ddd{3}x \ddd{3}y\\
    \leq h^3\tr(\gamma-\alpha\overline{\alpha})
    \|(\alpha_\psi\overline{\alpha_\psi})(\cdot,\cdot)\|_\infty
    \|V\|_1 \lesssim h^2\,.
  \end{multline}
  For the remaining terms we use
  \eqref{twonorm} with $\sigma = \xi$ and the Cauchy-Schwarz inequality to obtain
  \begin{align*}
    &\int_{\mathbb{R}^6}
      (\gamma-\alpha\overline{\alpha})(x,x) \left| (\alpha_\psi\overline{\xi})(y,y)
      V\big((x-y)/h\big)\right|\ddd{3}x \ddd{3}y\\
    &\qquad\lesssim h^{-1}
    \int_{\mathbb{R}^9} (\gamma-\alpha\overline{\alpha})(x,x) \|
      \xi (\cdot, y)\|_2
     \left| V\big((x-y)/h\big)\right|\ddd{3}x \ddd{3}y\\
    &\qquad\lesssim h^{-1} \|\xi\|_2 \|V(\cdot
    /h)\|_2\tr(\gamma-\alpha\overline{\alpha})  \lesssim h^2.
  \end{align*}

  We now turn  to the estimates \eqref{eq:lower_bound:ex:2} and
  \eqref{eq:lower_bound:dir:2}.
 The difference of  the exchange terms in  \eqref{eq:lower_bound:ex:2} is bounded 
  by
  \begin{equation*}
     \|V\|_\infty \| \alpha\overline{\alpha} -   \alpha_\psi\overline{\alpha_\psi} \|_2  \| \alpha\overline{\alpha}  +   \alpha_\psi\overline{\alpha_\psi} \|_2  \,.
  \end{equation*}
  The $2$-norm of $\alpha\overline{\alpha} -   \alpha_\psi\overline{\alpha_\psi}$ can be bounded by the $3/2$-norm, which in turn is bounded by $O(h)$ according to \eqref{eq:d_alpha_2}. Moreover, $\|  \alpha_\psi\overline{\alpha_\psi} \|_2 = \|\alpha_\psi\|_4^2  \lesssim h^{1/2}$ by \eqref{eq:alpha_n}, proving \eqref{eq:lower_bound:ex:2}.

  For the direct term we insert the decomposition $\alpha = \alpha_\psi + \xi$ into the
  difference in \eqref{eq:lower_bound:dir:2}, yielding $15$
  terms. However, due to symmetry, it suffices to estimate the
  following $5$ terms
  \begin{subequations}
    \begin{align}
      \label{eq:lower_bound:dir:2:xxxx}
      &\int_{\mathbb{R}^6}
        (\xi\overline{\xi})(x,x)(\xi\overline{\xi})(y,y)
        \left|V\big((x-y)/h\big)\right|\ddd{3}x \ddd{3}y,\\
      \label{eq:lower_bound:dir:2:xxaa}
      &\int_{\mathbb{R}^6}
        (\xi\overline{\xi})(x,x)(\alpha_\psi\overline{\alpha_\psi})(y,y)
        \left|V\big((x-y)/h\big)\right|\ddd{3}x \ddd{3}y, \\
      \label{eq:lower_bound:dir:2:xaaa}
      &\int_{\mathbb{R}^6}
        (\xi\overline{\alpha_\psi})(x,x)(\alpha_\psi\overline{\alpha_\psi})(y,y)
        \left|V\left( (x-y)/h\right)\right|\ddd{3}x \ddd{3}y,\\
      \label{eq:lower_bound:dir:2:xxxa}
      &\int_{\mathbb{R}^6}
        (\xi\xi)(x,x)(\xi\overline{\alpha_\psi})(y,y)
        \left|V\left( (x-y)/h\right)\right|\ddd{3}x \ddd{3}y,\\
      \label{eq:lower_bound:dir:2:xaxa}
      &\int_{\mathbb{R}^6}
        (\xi\overline{\alpha_\psi})(x,x)(\xi\overline{\alpha_\psi})(y,y)
        \left|V\left( (x-y)/h\right)\right|\ddd{3}x \ddd{3}y.
    \end{align}
  \end{subequations}
  We begin with  \eqref{eq:lower_bound:dir:2:xxxx}. Obviously
  \begin{equation*}
    \int_{\mathbb{R}^6}
    (\xi\overline{\xi})(x,x)(\xi\overline{\xi})(y,y)
    \left|V\big((x-y)/h\big)\right|\ddd{3}x \ddd{3}y \leq   \|V\|_\infty \big[\tr (\xi\overline{\xi})\big]^2 \lesssim h^2.
  \end{equation*}
  For \eqref{eq:lower_bound:dir:2:xxaa} we obtain with the help of \eqref{eq:sup_alpha2}
  \begin{equation*}
    \int_{\mathbb{R}^6}
    (\xi\overline{\xi})(x,x)(\alpha_\psi\overline{\alpha_\psi})(y,y)
    \left|V\big((x-y)/h\big)\right|\ddd{3}x \ddd{3}y \leq
    \tr(\xi\overline{\xi})
    \|(\alpha_\psi\overline{\alpha_\psi})(\cdot,\cdot)\|_\infty h^3
    \|V\|_1  \lesssim h^{2}.
  \end{equation*}
  For the last three terms we invoke equation \eqref{twonorm} from
  Lemma~\ref{lemma:inequalpha} with $\sigma = \xi$ and the Cauchy-Schwarz inequality. 
  For  \eqref{eq:lower_bound:dir:2:xaaa} this gives
  \begin{align*}
    &\int_{\mathbb{R}^6}
    (\xi\overline{\alpha_\psi})(x,x)(\alpha_\psi\overline{\alpha_\psi})(y,y)
    \left|V\left( (x-y)/h\right)\right|\ddd{3}x \ddd{3}y \\
    &\qquad\lesssim h^{-1} \int_{\mathbb{R}^6}
    \|\xi(x,\cdot)\|_2 (\alpha_\psi\overline{\alpha_\psi})(y,y)
    \left|V\left( (x-y)/h\right)\right|\ddd{3}x \ddd{3}y  \\
    &\qquad\lesssim h^{-1} \|V(\cdot /h)\|_1 \|
    \xi\|_2 \|\alpha_\psi \overline{\alpha_\psi}(\cdot, \cdot)\|_2 \,.
  \end{align*}
  The desired bound $O(h^{3/2})$ then follows from the fact that the last factor $\|\alpha_\psi \overline{\alpha_\psi}(\cdot, \cdot)\|_2$
  on the right  side is of order $O(h^{-1})$. To see this, we write 
  \begin{align*}
    \|\alpha_\psi \overline{\alpha_\psi}(\cdot, \cdot)\|_2^2 &=
    \int_{\mathbb{R}^9} |\alpha_\psi(x,y)|^2 |\alpha_\psi(x,z)|^2
    \ddd{3}x\ddd{3}y\ddd{3}z\\
    &= h^{-8}\int_{\mathbb{R}^9} |\alpha_0\big((x-y)/h\big)|^2
    |\alpha_0\big((x-z)/h\big)|^2|\psi\big((x+y)/2\big)|^2|\psi\big((x+z)/2\big)|^2
    \ddd{3}x\ddd{3}y\ddd{3}z.
  \end{align*}
  Changing to the variables $r = x-y$, $s = x-z$ and $x$ and using
  Cauchy-Schwarz in $x$, we indeed obtain
  \begin{align*}
    \|\alpha_\psi \overline{\alpha_\psi}(\cdot, \cdot)\|_2^2 =
    h^{-8}\int_{\mathbb{R}^9} |\alpha_0\big(r/h\big)|^2
    |\alpha_0\big(s/h\big)|^2|\psi(x-r/2)|^2|\psi(x-s/2)|^2
    \ddd{3}x\ddd{3}r\ddd{3}s \leq h^{-2} \|\alpha_0\|_2^2
    \|\psi\|_4^4.
  \end{align*}
For   \eqref{eq:lower_bound:dir:2:xxxa} we get 
  \begin{align*}
    &\int_{\mathbb{R}^6}
    (\xi\overline{\xi})(x,x)(\xi\overline{\alpha_\psi})(y,y)
    \left|V\big((x-y)/h\big)\right|\ddd{3}x \ddd{3}y \\
    &\qquad\lesssim h^{-1}  \int_{\mathbb{R}^6}
    (\xi\overline{\xi})(x,x) \|\xi(y,\cdot)\|_2 \left|V\left(
        (x-y)/h\right)\right|\ddd{3}x \ddd{3}y
    \\
    &\qquad\lesssim h^{-1}  \|V(\cdot /h)\|_2
    \|\xi\|_2^3   \lesssim h^2
  \end{align*}
  and for \eqref{eq:lower_bound:dir:2:xaxa}
  \begin{align*}
    &\int_{\mathbb{R}^6}
    (\xi\overline{\alpha_\psi})(x,x)(\xi\overline{\alpha_\psi})(y,y)
    \left|V\big((x-y)/h\big)\right|\ddd{3}x \ddd{3}y \\
    &\qquad\lesssim h^{-2} 
    \int_{\mathbb{R}^6} \|\xi(x,\cdot)\|_2 \|\xi(y,\cdot)\|_2
    \left|V\left( (x-y)/h\right)\right|\ddd{3}x \ddd{3}y  \\
    &\qquad\lesssim h^{-2}  \|\xi\|^2_2
    \|V(\cdot/h)\|_1  \lesssim h^2.
  \end{align*}
  This concludes the proof of
  \eqref{eq:lower_bound:ex} and \eqref{eq:lower_bound:dir}.

  \appendix

  \section{Appendix: Proof of Lemma~\ref{lemma:inequalpha}}
  \label{sec:alpha_4}

\begin{proof}[Proof of Lemma~\ref{lemma:inequalpha}, Part II]
   We first prove \eqref{eq:alpha_n} and \eqref{eq:nabla_alpha_n}.  For $n \in 2 \mathbb{N}$, we can write 
  \begin{equation}
    \label{eq:tr_alpha_n}
    \tr\big((\alpha_\psi\overline{\alpha_\psi})^{n/2}\big)
    =
    \int_{\mathbb{R}^{3n}}
    \alpha_\psi(x_1,x_2)\overline{\alpha_\psi(x_2,x_3)}\cdots\alpha_\psi(x_{n-1},x_{n})\overline{\alpha_\psi(x_{n},x_1)}
    \ddd{3} x_1\cdots \ddd{3} x_n.
  \end{equation}
  We switch to the following coordinates
  \begin{equation}
    \label{eq:coordinates}
    \begin{split}
      X &= \frac{1}{n}\sum_{k=1}^n x_k \\
      r_k &= x_{k+1}-x_k,\qquad k = 1,\dotsc,n-1.
    \end{split}
  \end{equation}
  It is easy to see that the corresponding Jacobi determinant is
  equal to $1$.  Moreover, we can recover the original coordinates via
  \begin{align*}
    x_1 &= X - \frac{1}{n} \sum_{i=1}^{n-1}(n-i)r_i, \\
    x_{k+1} &= x_k + r_k,
  \end{align*}
  i.e.
  \begin{align*}
    x_k = X + s_k(r_1,\dotsc,r_{n-1})
  \end{align*}
  for some  linear functions $s_k$.  We therefore obtain for the integral in
  \eqref{eq:tr_alpha_n}
  \begin{align*}
    \|\alpha_\psi\|_n^n &= h^{-2n}\int_{\mathbb{R}^{3n}}
    \psi\bigl(X+ \tfrac 12( s_1+s_2) \bigr) \cdots
    \overline{\psi\bigl(X+ \tfrac 12(s_n + s_1)\bigr)} \\
    &\qquad\qquad\qquad \times \alpha_0(r_1/h)\cdots \overline{\alpha_0(r_{n}/h)}
    \ddd{3} X \ddd{3} r_1 \cdots\ddd{3} r_{n-1},
  \end{align*}
  where we introduced $r_n := -\sum_{k=1}^{n-1} r_k$.  H\"older's inequality in the $X$ variable then yields 
  \begin{align*}
    \| \alpha_\psi \|_n^n &\leq h^{-2n} \|\psi\|_n^n
    \int_{\mathbb{R}^{3(n-1)}} \left|\alpha_0(r_1/h)\cdots
      \overline{\alpha_0(r_{n}/h)}\right| \ddd{3} r_1 \cdots\ddd{3}
    r_{n-1}
    \\
    &= (2\pi)^{3(n-2)/2}h^{n-3}\|\psi\|_n^n\,
    \|\widehat{|\alpha_0|}\|_n^n\,,
  \end{align*}
  which is \eqref{eq:alpha_n}. 
  The same calculation with $\alpha_0$ replaced by
  $\nabla\alpha_0$ yields \eqref{eq:nabla_alpha_n}.

    Due to the symmetry $\alpha_\psi(x,y) = \alpha_\psi(y,x)$, we have
  \begin{align*}
    \tr
     \big( \Delta \alpha_\psi \overline{\alpha_\psi}\alpha_\psi\overline{\alpha_\psi}\big)
    &= \langle\alpha_\psi
    \overline{\alpha_\psi}\alpha_\psi, \Delta_x
    \alpha_\psi \rangle_{L^2(\mathbb{R}^6)} = \langle\alpha_\psi
    \overline{\alpha_\psi}\alpha_\psi, \tfrac 12 
    (\Delta_x+\Delta_y)\alpha_\psi \rangle_{L^2(\mathbb{R}^6)}\\
    &= \langle\alpha_\psi
    \overline{\alpha_\psi}\alpha_\psi, ( \tfrac 14 
    \Delta_X +\Delta_r) \alpha_\psi \rangle_{L^2(\mathbb{R}^6)} \,.
  \end{align*}
    Using the coordinates
  \eqref{eq:coordinates}, for which we have in the case of $n=4$
  \begin{align*}
    \frac{x_1+x_2}{2} &= X - s & \frac{x_3+x_4}{2} &= X + s &
    s(r_1, r_2, r_3) &= \frac{r_1 + 2r_2 + r_3}{4} \\
    \frac{x_2+x_3}{2} &= X - t & \frac{x_1+x_4}{2} &= X + t & t(r_1,
    r_2, r_3) &= \frac{r_3 - r_1}{4}
  \end{align*}
  and rescaling $r_k \to hr_k$, $k=1,2,3$, we
  can therefore write
  \begin{align*}
    & \tr(-h^2 \Delta +
    E_b/2)\alpha_\psi\overline{\alpha_\psi}\alpha_\psi\overline{\alpha_\psi}
    \\ & = h\int_{\mathbb{R}^{12}} \psi(X-hs) \overline{\psi(X-ht)}
    \psi(X+hs)
    \overline{\psi(X+ht)}\\
    &\qquad\qquad \times \bigl[(-\Delta+E_b/2)\alpha_0(r_1)\bigr]
    \overline{\alpha_0(r_2)}\alpha_0(r_3)
    \overline{\alpha_0(-r_1-r_2-r_3)}
    \ddd{3} X \ddd{3} r_1\ddd{3} r_2\ddd{3} r_3\\
    &\quad -\frac{h^2}{4} \langle\alpha_\psi
    \overline{\alpha_\psi}\alpha_\psi,  
    \Delta_X  \alpha_\psi \rangle_{L^2(\mathbb{R}^6)}\,.
    \end{align*}
  This term has the form
  \begin{align*}
    h (2\pi)^3 \|\psi\|_4^4 \int_{\mathbb{R}^3}
    |\widehat{\alpha_0}(p)|^4(p^2+E_b/2) \ddd{3}p + A_1\,h^{2} +
    A_2\,h^2,
  \end{align*}
  where
  \begin{align*}
    A_1 &= -\frac{1}{4}  \langle\alpha_\psi
    \overline{\alpha_\psi}\alpha_\psi,  
    \Delta_X  \alpha_\psi \rangle_{L^2(\mathbb{R}^6)}  \,,
    \\
    A_2 &= h^{-1}\int_{\mathbb{R}^{12}} \int_0^1 \frac{\dd}{\dd
      \tau}\left( \psi(X-\tau hs) \overline{\psi(X-\tau ht)}
      \psi(X+\tau hs) \overline{\psi(X+\tau ht)}\right)\dd{\tau} \\
    &\qquad\qquad \times \bigl[(-\Delta+E_b/2)\alpha_0(r_1)\bigr]
    \overline{\alpha_0(r_2)}\alpha_0(r_3)
    \overline{\alpha_0(-r_1-r_2-r_3)} \ddd{3} X \ddd{3} r_1\ddd{3}
    r_2\ddd{3} r_3.
  \end{align*}
  Using integration by parts, we can bound  $A_1$ as
  \begin{align*}
    |A_1| & = \frac{1}{4} \left| \left\langle \nabla_X \big(\alpha_\psi
    \overline{\alpha_\psi}\alpha_\psi\big) , \nabla_X \alpha_\psi \right \rangle  \right| \\& = \frac 14   \left| \left\langle \big( \nabla_X \alpha_\psi\big)
    \overline{\alpha_\psi}\alpha_\psi +  \alpha_\psi \big( \nabla_X
    \overline{\alpha_\psi}\big)\alpha_\psi   +   \alpha_\psi
    \overline{\alpha_\psi} \big( \nabla_X\alpha_\psi\big)  , \nabla_X \alpha_\psi \right \rangle  \right| \\ &    \leq  \frac 34 \|\nabla_X \alpha_\psi \|_2^2 \|\alpha_\psi\|_\infty^2  \lesssim \|\nabla\psi\|_2^4 \,,
  \end{align*}
  where the last inequality follows from $\|\alpha_\psi\|_\infty \leq \|\alpha_\psi\|_6$, which is $\lesssim h^{1/2} \|\psi\|_6$ as shown above. 
  
  To estimate $A_2$, we carry out the derivative in $\tau$ and subsequently use H\"older's inequality for the $X$ integration to obtain
    \begin{equation*}
    |A_2| \leq \|\nabla \psi\|_2 \|\psi\|_6^3
    \int_{\mathbb{R}^{9}}
    \big(|s|+|t|\big)\big|(V\alpha_0)(r_1)
    \alpha_0(r_2) \alpha_0(r_3)
    \alpha_0(-r_1-r_2-r_3)\big|
    \ddd{3} r_1\ddd{3} r_2\ddd{3} r_3.
  \end{equation*}
Here we have also used that $(-\Delta + E_b/2)\alpha_0 =
 - \tfrac 12 V \alpha_0$.
  We now note that $|s|+|t| \leq \frac 12 |r_1+r_2+r_3| + \frac 12 |r_2| + \frac 12 |r_3|$. We plug in this bound in the integrand and use Cauchy-Schwarz for the $r_2$ integration in the case of the terms 
  $|r_1+r_2+r_3|$ and $|r_2|$, and for the $r_3$ integration in the case
  of $|r_3|$. This yields 
  \begin{equation*}
    |A_2|\leq
    \frac 32\|\nabla \psi\|_2 \|\psi\|_6^3
    \|V\alpha_0\|_1\, \|\alpha_0\|_1\, \|\alpha_0\|_2\,
    \big\||\cdot|\alpha_0\big\|_2 \,.
  \end{equation*}
 The desired result then follows from the Sobolev inequality $\|\psi\|_6 \lesssim \|\nabla
  \psi\|_2$.
  This concludes the proof of Lemma~\ref{lemma:inequalpha}.
\end{proof}

\section*{Acknowledgments}

Partial financial support from the DFG grant GRK 1838, as well as the Austrian Science Fund (FWF), project Nr. P 27533-N27 (R.S.), is gratefully acknowledged.

\bibliographystyle{amsplain}

\end{document}